\PassOptionsToPackage{unicode}{hyperref}
\PassOptionsToPackage{hyphens}{url}
\PassOptionsToPackage{dvipsnames,svgnames,x11names}{xcolor}

\documentclass[12pt]{article}

\usepackage{amsmath,amssymb}
\usepackage{iftex}
\ifPDFTeX
  \usepackage[T1]{fontenc}
  \usepackage[utf8]{inputenc}
  \usepackage{textcomp} 
\else 
  \usepackage{unicode-math}
  \defaultfontfeatures{Scale=MatchLowercase}
  \defaultfontfeatures[\rmfamily]{Ligatures=TeX,Scale=1}
\fi
\usepackage{lmodern}
\ifPDFTeX\else  
 
\fi
\IfFileExists{upquote.sty}{\usepackage{upquote}}{}
\IfFileExists{microtype.sty}{
  \usepackage[]{microtype}
  \UseMicrotypeSet[protrusion]{basicmath} 
}{}
\makeatletter
\@ifundefined{KOMAClassName}{
  \IfFileExists{parskip.sty}{
    \usepackage{parskip}
  }{
    \setlength{\parindent}{0pt}
    \setlength{\parskip}{6pt plus 2pt minus 1pt}}
}{
  \KOMAoptions{parskip=half}}
\makeatother
\usepackage{xcolor}
\makeatletter
\ifx\paragraph\undefined\else
  \let\oldparagraph\paragraph
  \renewcommand{\paragraph}{
    \@ifstar
      \xxxParagraphStar
      \xxxParagraphNoStar
  }
  \newcommand{\xxxParagraphStar}[1]{\oldparagraph*{#1}\mbox{}}
  \newcommand{\xxxParagraphNoStar}[1]{\oldparagraph{#1}\mbox{}}
\fi
\ifx\subparagraph\undefined\else
  \let\oldsubparagraph\subparagraph
  \renewcommand{\subparagraph}{
    \@ifstar
      \xxxSubParagraphStar
      \xxxSubParagraphNoStar
  }
  \newcommand{\xxxSubParagraphStar}[1]{\oldsubparagraph*{#1}\mbox{}}
  \newcommand{\xxxSubParagraphNoStar}[1]{\oldsubparagraph{#1}\mbox{}}
\fi
\makeatother

\usepackage{longtable,booktabs,array}
\usepackage{calc} % for calculating minipage widths
\usepackage{etoolbox}
\makeatletter
\patchcmd\longtable{\par}{\if@noskipsec\mbox{}\fi\par}{}{}
\makeatother
\IfFileExists{footnotehyper.sty}{\usepackage{footnotehyper}}{\usepackage{footnote}}
\makesavenoteenv{longtable}
\usepackage{graphicx}
\makeatletter
\def\maxwidth{\ifdim\Gin@nat@width>\linewidth\linewidth\else\Gin@nat@width\fi}
\def\maxheight{\ifdim\Gin@nat@height>\textheight\textheight\else\Gin@nat@height\fi}
\makeatother
\setkeys{Gin}{width=\maxwidth,height=\maxheight,keepaspectratio}
\makeatletter
\def\fps@figure{htbp}
\makeatother

\makeatletter
\@ifpackageloaded{caption}{}{\usepackage{caption}}
\AtBeginDocument{%
\ifdefined\contentsname
  \renewcommand*\contentsname{Table of contents}
\else
  \newcommand\contentsname{Table of contents}
\fi
\ifdefined\listfigurename
  \renewcommand*\listfigurename{List of Figures}
\else
  \newcommand\listfigurename{List of Figures}
\fi
\ifdefined\listtablename
  \renewcommand*\listtablename{List of Tables}
\else
  \newcommand\listtablename{List of Tables}
\fi
\ifdefined\figurename
  \renewcommand*\figurename{Figure}
\else
  \newcommand\figurename{Figure}
\fi
\ifdefined\tablename
  \renewcommand*\tablename{Table}
\else
  \newcommand\tablename{Table}
\fi
}
\@ifpackageloaded{float}{}{\usepackage{float}}
\floatstyle{ruled}
\@ifundefined{c@chapter}{\newfloat{codelisting}{h}{lop}}{\newfloat{codelisting}{h}{lop}[chapter]}
\floatname{codelisting}{Listing}

\makeatother
\makeatletter
\@ifpackageloaded{caption}{}{\usepackage{caption}}
\@ifpackageloaded{subcaption}{}{\usepackage{subcaption}}
\makeatother

\ifLuaTeX
  \usepackage{selnolig}  % disable illegal ligatures
\fi
\usepackage[]{natbib}
\usepackage{bookmark}

\IfFileExists{xurl.sty}{\usepackage{xurl}}{} % add URL line breaks if available
\hypersetup{
  pdftitle={Title},
  pdfauthor={Author 1; Author 2},
  pdfkeywords={3 to 6 keywords, that do not appear in the title},
  colorlinks=true,
  linkcolor={blue},
  filecolor={Maroon},
  citecolor={Blue},
  urlcolor={Blue},
  pdfcreator={LaTeX via pandoc}}

\usepackage{mathtools}
\usepackage{amsthm}
\usepackage[ruled,vlined]{algorithm2e}
\usepackage{tikz}
\usepackage{tabularx}
\usepackage{threeparttable}
\usepackage{enumitem}
\usepackage[section]{placeins}
\usepackage{cleveref}

\theoremstyle{plain}
\newtheorem{theorem}{Theorem}[section]
\newtheorem{proposition}[theorem]{Proposition}
\newtheorem{lemma}[theorem]{Lemma}
\newtheorem{corollary}[theorem]{Corollary}

\theoremstyle{definition}

\theoremstyle{remark}

\crefname{theorem}{theorem}{theorems}
\Crefname{theorem}{Theorem}{Theorems}
\crefname{proposition}{proposition}{propositions}
\Crefname{proposition}{Proposition}{Propositions}
\crefname{lemma}{lemma}{lemmas}
\Crefname{lemma}{Lemma}{Lemmas}
\crefname{corollary}{corollary}{corollaries}
\Crefname{corollary}{Corollary}{Corollaries}
\crefname{assumption}{assumption}{assumptions}
\Crefname{assumption}{Assumption}{Assumptions}
\crefname{remark}{remark}{remarks}
\Crefname{remark}{Remark}{Remarks}

\newcommand{\Unif}{\mathrm{Unif}}

\DeclareMathOperator{\Var}{Var}

\newcommand{\sel}{\mathrm{sel}}
\newcommand{\Bel}{\mathrm{Bel}}
\newcommand{\Pl}{\mathrm{Pl}}

\newcommand{\anon}{1}

\begin{document}

\def\spacingset#1{\renewcommand{\baselinestretch}%
{#1}\small\normalsize} \spacingset{1}

%%%%%%%%%%%%%%%%%%%%%%%%%%%%%%%%%%%%%%%%%%%%%%%%%%%%%%%%%%%%%%%%%%%%%%%%%%%%%%

\if1\anon
{
  \title{\bf Possibilistic Inferential Models for Post-Selection Inference in High-Dimensional Linear Regression}
  \author{Lin Yaohui 1\thanks{
     \textit{}}\hspace{.2cm}\\
    South China Normal University\\
}
  \maketitle
} \fi

\if0\anon
{
  \bigskip
  \bigskip
  \bigskip
  \begin{center}
    {\LARGE\bf Possibilistic Inferential Models for Post-Selection Inference in High-Dimensional Linear Regression}
\end{center}
  \medskip
} \fi

\bigskip
\begin{abstract}
Valid uncertainty quantification after model selection remains challenging in high-dimensional linear regression, especially within the possibilistic inferential model (PIM) framework. We develop possibilistic inferential models for post-selection inference based on a \emph{regularized split possibilistic} construction (RSPIM) that combines generic high-dimensional selectors with PIM validification through sample splitting. A first subsample is used to select a sparse model; ordinary least-squares refits on an independent inference subsample yield classical $t/F$ pivots, which are then turned into consonant plausibility contours. In Gaussian linear models this leads to coordinatewise intervals with exact finite-sample strong validity conditional on the split and selected model, uniformly over all selectors that use only the selection data. We further analyze RSPIM in a sparse $p \gg n$ regime under high-level screening conditions, develop orthogonalized and bootstrap-based extensions for low-dimensional targets with high-dimensional nuisance, and study a maxitive multi-split aggregation that stabilizes inference across random splits while preserving strong validity. Simulations and a riboflavin gene-expression example show that calibrated RSPIM intervals are well behaved under both Gaussian and heteroskedastic errors and are competitive with state-of-the-art post-selection methods, while plausibility contours provide transparent diagnostics of post-selection uncertainty.
\end{abstract}

\vfill

\newpage
\spacingset{1.8} % DON'T change the spacing!

\section{Introduction}\label{sec:intro}

High-dimensional linear models with $p \gg n$ arise routinely in applications such as genomics, neuroscience, and macroeconomics. A common workflow is to first apply a regularized method (for example, the lasso) to screen variables and then fit an ordinary least-squares model on the selected support to obtain estimates and standard errors. While this pipeline is convenient and empirically effective for prediction, naive post-selection intervals that ignore the selection step can suffer from substantial and systematic under-coverage, especially when selection is aggressive or signals are weak. The resulting uncertainty quantification may be misleading in ways that are difficult to diagnose from standard output.

A substantial literature on post-selection inference (PSI) has been developed to address these issues. One strand constructs exact or nearly exact selective procedures by conditioning on the selection event, leading to truncated-Gaussian pivots and polyhedral confidence sets in linear models with lasso or related selectors; see, for example, the polyhedral selective-inference framework for lasso and forward stepwise regression. These methods deliver finite-sample guarantees but can be computationally demanding and may produce extremely wide or even infinite intervals in highly correlated or weak-signal regimes. Another strand constructs de-biased or de-sparsified estimators and uses asymptotic normality of corrected statistics to form confidence intervals after regularization \citep{ZhangZhang2014,javanmard2014confidence,vandeGeer2014optimal}. Such approaches scale well to large $p$ but rely on Gaussian or sub-Gaussian approximations, stability of sparsity patterns, and suitable design conditions.

Inferential models (IMs) and their possibilistic refinement (PIMs) offer a different starting point. In this framework, uncertainty about parameters is represented via belief and plausibility rather than probability measures on the parameter space, and inference is driven by calibration of predictive random sets or plausibility contours derived from model-based pivots \citep{martinliu2013,martinliu2014}. The resulting procedures are “prior-free’’ in the Bayesian sense while retaining strong frequentist validity guarantees and providing graphical diagnostics of how concentrated or diffuse the evidence is. In the possibilistic specialization, consonant plausibility contours summarize information in a single curve $\pi_Z(\theta)\in[0,1]$, with upper-level sets functioning as confidence regions.

Extending these IM/PIM constructions to high-dimensional problems with regularization has been identified as a key open direction. In low-dimensional parametric or semiparametric models without regularization, generic validification results show how to turn exact or asymptotically pivotal statistics into plausibility contours with strong validity properties. In high dimensions, however, the screening step is itself data-dependent and often difficult to encode within existing IM/PIM templates. This raises a natural question: can we decouple high-dimensional selection from possibilistic validification so that standard IM/PIM techniques can be applied without explicitly modeling the selection event?

\subsection{Positioning and relation to prior work}

This work lies at the intersection of PSI and possibilistic inference. On the PSI side, our proposal offers a complementary route to conditional and de-biased approaches. Rather than modeling the selection event to obtain truncated-Gaussian pivots, or constructing de-biased estimators with asymptotic normality guarantees, we focus on a split-and-refit architecture that restores a low-dimensional Gaussian linear model on an independent inference subsample. This makes exact or bootstrap-based validification feasible even when the original problem is high-dimensional and the selector is complex, and it yields selector-agnostic finite-sample guarantees expressed in terms of plausibility contours.

Within the IM/PIM literature, most existing constructions focus on low-dimensional parametric or semiparametric models without regularization, where the pivot or likelihood-ratio ranking is defined for a fixed data-generating model and does not involve a data-dependent selected support. Generic validification results convert fixed-dimensional pivots into plausibility contours but do not address the interaction between high-dimensional screening and subsequent uncertainty quantification. By treating the selector as a black box on a separate subsample and then validifying a refitted pivot on the inference subsample, our regularized split possibilistic inference (RSPIM) procedure provides a concrete high-dimensional instance of the PIM framework in the $p \gg n$ regime and illustrates how regularization and possibilistic validification can be combined in a modular way.

The remainder of the paper is organized as follows. Section~\ref{sec:method} introduces the RSPIM construction, including the split-and-refit scheme, the possibilistic validification step, orthogonalized extensions, and multi-split aggregation. Section~\ref{sec:theory} develops finite-sample and high-dimensional strong-validity results, together with orthogonalized and robust extensions. Section~\ref{sec:computation} discusses implementation details, including bootstrap validification and calibration. Section~\ref{sec:experiments} presents simulation studies on calibration, efficiency, robustness, and comparisons with state-of-the-art PSI methods. Section~\ref{sec:discussion} concludes with a brief discussion of limitations and directions for future work.

\subsection{Contributions}\label{sec:contributions}

Our main contributions can be summarized as follows.

\subsubsection{Selector-agnostic split possibilistic inference.}
We show how a classical split-and-refit scheme can be organized within the possibilistic IM/PIM framework to yield selector-agnostic, finite-sample valid post-selection inference. A random split partitions the data into selection and inference subsamples; on the inference subsample we refit ordinary least squares on the selected support and apply a generic PIM validification recipe to the resulting $t/F$ pivots. Under Gaussian homoskedastic errors and independent splitting, the resulting coordinatewise intervals enjoy finite-sample strong validity conditional on the selected model and hence unconditionally. These guarantees hold uniformly over all selectors that use only the selection data.

\subsubsection{High-dimensional strong validity under screening conditions.}
We embed this construction in a sparse high-dimensional linear model with parameter space $\mathcal B_{s_n}=\{\beta\in\mathbb R^{p_n}:\|\beta\|_0\le s_n\}$. Under high-dimensional assumptions (HD) and a screening property for the split-based selector, single-split RSPIM intervals are asymptotically strongly valid uniformly over $\mathcal B_{s_n}$ for active coordinates and can be oracle-equivalent in length when the selector is consistent. For lasso combined with stability selection we verify the screening condition under standard sub-Gaussian design, beta-min, and tuning assumptions.

\subsubsection{Orthogonalized and robust extensions.}
For problems with low-dimensional targets and high-dimensional nuisance components we develop orthogonalized and robust extensions of RSPIM. We integrate Neyman-orthogonal scores with cross-fitting into the validification step and formalize generic orthogonality, rate, and bootstrap conditions (O1)–(O4) under which the resulting plausibility contour is asymptotically strongly valid. In a partially linear regression with high-dimensional controls we verify these conditions under approximate sparsity and restricted-eigenvalue assumptions, combining the construction with wild-bootstrap validification to accommodate heteroskedastic and heavy-tailed errors.

\subsubsection{Empirical evaluation and diagnostics.}
Simulation studies and a riboflavin gene-expression example show that calibrated RSPIM intervals are well behaved under both Gaussian and non-Gaussian noise and competitive with de-biased lasso and polyhedral selective inference. Plausibility contours and maxitive aggregation provide intuitive graphical diagnostics of post-selection uncertainty, highlighting when apparently significant effects rest on fragile selection events.

\section{Method}\label{sec:method}

\subsection{Model, split, and estimands}\label{subsec:model-split}

We observe independent pairs $(Y_i,X_i) \in \mathbb{R} \times \mathbb{R}^p$, $i=1,\dots,n$, from the linear model
\[
Y = X\beta_0 + \varepsilon,\qquad X \in \mathbb{R}^{n\times p},\quad \varepsilon \sim N(0,\sigma^2 I_n),
\]
unless stated otherwise. Let $X_j$ denote the $j$th column of $X$ and write
\[
S_0 = \{j : \beta_{0j} \neq 0\}
\]
for the (unknown) true support.

We randomly split the index set $\{1,\dots,n\}$ into a selection subset $I_{\mathrm{sel}}$ and an inference subset $I_{\mathrm{inf}}$, with $I_{\mathrm{sel}}\cap I_{\mathrm{inf}}=\emptyset$ and $I_{\mathrm{sel}}\cup I_{\mathrm{inf}}=\{1,\dots,n\}$. Let $n_{\mathrm{sel}} = |I_{\mathrm{sel}}|$ and $n_{\mathrm{inf}} = |I_{\mathrm{inf}}|$.

A generic selector is a measurable map
\[
S : (Y_{I_{\mathrm{sel}}}, X_{I_{\mathrm{sel}}}) \mapsto \widehat S \subseteq \{1,\dots,p\},
\]
which uses only the selection data and returns a random subset $\widehat S$ of active coordinates. We assume that $S$ enforces a size cap $|\widehat S| \le k_{\max} \le n_{\mathrm{inf}}$, so that least-squares refitting on $I_{\mathrm{inf}}$ is well defined.

Inference is carried out solely on $I_{\mathrm{inf}}$. Conditional on the realized split and selected support, our basic targets are (i) the coordinates $\beta_j$ for $j \in \widehat S$ and (ii) low-dimensional linear contrasts $\theta_S = L^\top \beta_S$ with $S\subseteq\widehat S$ and $L\in\mathbb{R}^{|S|\times q}$ fixed and low-dimensional. In Section~\ref{sec:orth-full} we extend the construction to full-model coordinates $j\in\{1,\dots,p\}$, including those not selected into $\widehat S$.

\subsection{Refit and $F$-pivot}\label{subsec:refit}

On the inference subset $I_{\mathrm{inf}}$ we refit ordinary least squares on the selected variables. Let
\[
X_{\mathrm{inf}} = X_{I_{\mathrm{inf}}},\qquad
Y_{\mathrm{inf}} = Y_{I_{\mathrm{inf}}},\qquad
X_S = (X_{\mathrm{inf},j})_{j\in S}
\]
denote the inference-sample design and response, and write $S = \widehat S$ and $d = |S|$. The refitted least-squares estimator and residual variance are
\[
\widehat\beta_S
= (X_S^\top X_S)^{-1}X_S^\top Y_{\mathrm{inf}},\qquad
\widehat\sigma^2
= \frac{\|Y_{\mathrm{inf}} - X_S\widehat\beta_S\|_2^2}{n_{\mathrm{inf}}-d},
\]
with covariance
\[
\widehat\Sigma_S
= \widehat\sigma^2 (X_S^\top X_S)^{-1}.
\]
Let $v_{jj}$ denote the $(j,j)$ entry of $(X_S^\top X_S)^{-1}$ and $e_j$ the $j$th standard basis vector in $\mathbb{R}^d$.

For each coordinate $j\in S$, the usual $t$-statistic
\[
t_j
= \frac{\widehat\beta_j - \beta_{0j}}{\widehat\sigma\sqrt{v_{jj}}}
\]
has a $t_\nu$ distribution with degrees of freedom $\nu = n_{\mathrm{inf}}-d$ under $\beta_{0j}$, conditional on the split and selected support, provided the Gaussian model holds on $I_{\mathrm{inf}}$. Likewise, for a linear contrast $\theta_S = L^\top\beta_S$ the likelihood-ratio statistic for testing $H_0:\theta_S=\theta_{0,S}$ has an $F_{q,\nu}$ distribution. We collect these pivot properties formally in Lemma~\ref{lem:selector-pivot} in the theory section.

In what follows we treat these $t/F$ statistics as selector-agnostic pivots: conditional on the split and selected support they have the same null distribution regardless of how $\widehat S$ was obtained, as long as the selector uses only the selection data and the refit is performed on $I_{\mathrm{inf}}$. The possibilistic step then turns these pivots into plausibility contours and intervals.

\subsection{Possibilistic validification and plausibility contours}\label{subsec:validification}

We briefly recall the possibilistic framework. Let $\Theta$ denote the parameter space and $Z$ denote the observed data. An inferential random set is a random subset $\mathcal{S}(Z)$ of $\Theta$. The associated plausibility function on assertions $A \subseteq \Theta$ is
\[
\mathrm{Pl}_Z(A) = P\{\mathcal{S}(Z)\cap A \neq \emptyset \mid Z\},
\]
and for singleton assertions $\{\theta\}$ this reduces to the plausibility contour
\[
\pi_Z(\theta) = \mathrm{Pl}_Z(\{\theta\}) \in [0,1],\qquad \theta\in\Theta.
\]
An inferential model is called consonant if the upper-level sets
\[
C_Z(\alpha) = \{\theta : \pi_Z(\theta) \ge \alpha\}
\]
are nested in $\alpha$; in this case they can be interpreted as a family of plausibility regions indexed by $\alpha$.

A key property for frequentist uncertainty quantification is \emph{strong validity}: the contour $\pi_Z(\cdot)$ is strongly valid if
\begin{equation}
\label{eq:strong-validity}
\sup_{\theta_0 \in \Theta}
P_{\theta_0}\bigl\{\pi_Z(\theta_0) \le u\bigr\} \le u,
\qquad u \in [0,1].
\end{equation}
As a consequence, the upper-level sets $C_Z(\alpha) = \{\theta : \pi_Z(\theta) \ge \alpha\}$ are $(1-\alpha)$ confidence sets with exact or conservative frequentist coverage uniformly over $\theta_0$.

The RSPIM construction turns pivots into strongly valid contours via a probability integral transform. Given a scalar pivot $T(Z,\theta)$ with continuous cumulative distribution function $F_\theta$ under $P_\theta$, we define
\begin{equation}
\label{eq:plaus-contour}
U_\theta(Z) = F_\theta\{T(Z,\theta)\},\qquad
\pi_Z(\theta) = 1 - \bigl|2U_\theta(Z) - 1\bigr|.
\end{equation}
Under $P_{\theta_0}$, $U_{\theta_0}(Z)$ is $\Unif(0,1)$, and the transformation $u \mapsto 1 - |2u-1|$ preserves this uniformity, so $\pi_Z(\theta_0)\sim\Unif(0,1)$ and the strong-validity property~\eqref{eq:strong-validity} holds with equality. Detailed statements, together with a one-dimensional illustration, are collected in Appendix~A.

For later use with bootstrap pivots we also note that the construction remains approximately valid when $F_\theta$ is replaced by a consistent estimator $\widehat F_{n,\theta}$. If the Kolmogorov distance between the true and estimated pivot distributions vanishes uniformly in $\theta$, then the resulting plausibility contour is asymptotically strongly valid; see Appendix~A for a precise statement and proof. In practice this allows us to plug in a wild-bootstrap approximation to $F_\theta$ in heteroskedastic or non-Gaussian settings.

\subsection{Application to the RSPIM pivot}

Returning to the Gaussian linear model of Section~\ref{subsec:model-split}, the refit step in Section~\ref{subsec:refit} produces, conditional on the split and selected support, a $t$-pivot for each coordinate $j\in S$ and an $F$-pivot for linear contrasts $\theta_S$. Applying the recipe~\eqref{eq:plaus-contour} to these pivots yields consonant plausibility contours for $\theta_S$ and for each coordinate $\beta_j$, $j\in S$.

For a fixed coordinate $j\in S$, the $(1-\alpha)$ upper-level set of the resulting contour reduces to the usual $t$-interval
\[
C_{j,n}(\alpha)
= \bigl[\widehat\beta_j \pm t_{1-\alpha/2,\nu}\,\widehat\sigma \sqrt{v_{jj}}\bigr],
\qquad j\in S,
\]
and similarly for linear contrasts based on the $F$-pivot. Thus, in the Gaussian fixed-design case, RSPIM does not change the numerical form of the post-selection intervals: it embeds them in a consonant plausibility contour with strong-validity guarantees and provides a unified possibilistic framework for post-selection inference and diagnostics. The theory section shows that this strong validity holds conditional on the split and selected support and hence unconditionally, uniformly over all selectors that use only the selection data.

\subsection{Orthogonalized extension (full-model coordinates)}\label{sec:orth-full}

The construction above targets only the coordinates selected into $S = \widehat S$. To obtain inference for any coordinate $j\in\{1,\dots,p\}$, including those not in $\widehat S$, we embed the procedure in an orthogonal-score framework.

Fix $j\in\{1,\dots,p\}$. On the inference data we regress $X_j$ on $X_{S\setminus\{j\}}$ and $Y$ on $X_{S\setminus\{j\}}$ to obtain residuals $\widetilde x_j$ and $\widetilde y$:
\[
\widetilde x_j = X_j - X_{S\setminus\{j\}} \widehat\gamma_j,
\qquad
\widetilde y = Y - X_{S\setminus\{j\}} \widehat\delta,
\]
where $\widehat\gamma_j$ and $\widehat\delta$ are fitted on $I_{\mathrm{inf}}$ (or via cross-fitting in more general models). The resulting score
\[
U_j(\beta_j) = \widetilde x_j^\top (\widetilde y - \widetilde x_j \beta_j)
\]
is Neyman-orthogonal with respect to the high-dimensional nuisance parameters indexed by $S\setminus\{j\}$. Under standard moment and design conditions, the studentized version of $U_j(\beta_j)$ yields an asymptotically pivotal $t$-statistic, which we then validify using the same recipe~\eqref{eq:plaus-contour}. This produces plausibility contours and intervals for $\beta_j$ that remain valid even when $j\notin\widehat S$. Section~\ref{sec:orthogonalized} formalizes the required orthogonality, rate, and bootstrap conditions and establishes asymptotic strong validity in a general orthogonalized setting.

\subsection{Multi-split aggregation (maxitive)}\label{subsec:multisplit}

The basic RSPIM construction uses a single random split. To stabilize results with respect to the split, we also consider a maxitive multi-split aggregation.

Let $r = 1,\dots,R$ index independent splits of the data into selection and inference parts, and let $\pi_{n,j}^{(r)}(\cdot)$ and $C_{j,n}^{(r)}(\alpha)$ denote the split-wise contour and interval for coordinate $j$ on split $r$. The maxitive aggregate contour is
\[
\pi_{n,j}^{\max}(\beta_j)
= \max_{1\le r\le R} \pi_{n,j}^{(r)}(\beta_j),
\]
and the corresponding union-of-splits interval is
\[
C_{j,n}^\cup(\alpha)
= \{\beta_j : \pi_{n,j}^{\max}(\beta_j) > \alpha\}
= \bigcup_{r=1}^R C_{j,n}^{(r)}(\alpha).
\]
As shown in Proposition~\ref{prop:maxitive} in Section~\ref{sec:multi}, this maxitive aggregation preserves strong validity. In practice we treat $C_{j,n}^\cup(\alpha)$ as a conservative stability device that highlights sensitivity to the random split, rather than as our primary efficiency mechanism. The single-split interval $C_{j,n}(\alpha)$ remains the main object for efficiency comparisons.

\medskip
\noindent\emph{Optional data carving (approximate).}
If desired, a fraction of the selection data can be re-used for refitting by ``carving'' it into the inference set. This breaks the exact independence between selection and inference and, consequently, the finite-sample pivotal arguments in Section~\ref{subsec:refit}. We therefore regard carving as an approximate option and flag it explicitly whenever it is used; further discussion appears in Section~\ref{sec:computation}.

\section{Theory}\label{sec:theory}

This section develops theoretical guarantees for RSPIM. We begin with finite-sample results
for a single split in a Gaussian linear model, then embed the construction in a high-dimensional
sparse regime, extend it to orthogonalized and robust settings, and finally analyze the multi-split
maxitive aggregator. Throughout we let
\[
\mathcal{G} = \sigma(I_{\mathrm{sel}}, I_{\mathrm{inf}}, X, \widehat S)
\]
denote the $\sigma$-field generated by the split, the full design, and the selected support. All results are stated conditionally on $\mathcal{G}$ and therefore apply to any selection rule that uses only the selection data $I_{\mathrm{sel}}$. Detailed proofs and technical conditions appear in Appendix~\ref{app:theory-proofs}.

\subsection{Finite-sample post-selection strong validity}\label{sec:finite-sample}

We first consider the Gaussian linear model with sample splitting and refitting on the inference
subset. Under Assumptions~\textup{(L1)}–\textup{(L2)} (stated in Appendix~\ref{app:theory-proofs}), which encode a fixed-design Gaussian linear model on the inference sample together with an independent split-and-select step, the following results establish that RSPIM intervals enjoy strong validity conditional on $\mathcal{G}$, uniformly over all selectors that use only the selection sample.

\begin{lemma}[Selector-agnostic $t/F$ pivots]\label{lem:selector-pivot}
Under Assumptions~\textup{(L1)}–\textup{(L2)}, for any $j \in \widehat S$ we have
$t_j \sim t_\nu$ and $F \sim F_{d,\nu}$ conditional on $\mathcal{G}$, where
$d = |\widehat S|$ and $\nu = n_{\mathrm{inf}} - d$. In particular, the conditional distribution
of $(t_j,F)$ is the same for any selection rule that uses only the selection sample.
\end{lemma}

\begin{theorem}[Finite-sample post-selection strong validity]\label{thm:finite-sample-strong}
Under Assumptions~\textup{(L1)}–\textup{(L2)}, let $\pi_Z(\cdot)$ be the plausibility contour
constructed from the refitted $t/F$ pivot. For any coordinate $j \in \widehat S$, the single-split RSPIM interval $C_{j,n}(\alpha)$ satisfies
\[
P_{\beta_0}\bigl\{\beta_{0j} \notin C_{j,n}(\alpha) \,\big|\, \mathcal{G}\bigr\} \le \alpha,
\quad\text{and hence}\quad
\sup_{\beta_0} P_{\beta_0}\bigl\{\beta_{0j} \notin C_{j,n}(\alpha)\bigr\} \le \alpha.
\]
Moreover, the associated plausibility contour $\pi_Z(\cdot)$ is strongly valid:
\[
\sup_{\theta_0} P_{\theta_0}\bigl\{ \pi_Z(\theta_0) \le u \bigr\} \le u, \qquad u \in [0,1].
\]
\end{theorem}

The next result makes explicit that this finite-sample guarantee is uniform over all split-based
selectors that use only the selection data.

\begin{proposition}[Selector-uniform finite-sample strong validity]\label{prop:selector-uniform}
Let $\mathcal{S}$ be the collection of all measurable selectors
$\mathcal{S}:(Y_{I_{\mathrm{sel}}},X_{I_{\mathrm{sel}}})\mapsto \widehat S \subseteq \{1,\dots,p\}$ that use only
the selection data and satisfy $|\widehat S|<n_{\mathrm{inf}}$ almost surely. For each
$\mathcal{S}\in\mathcal{S}$, let $C^{(\mathcal{S})}_{j,n}(\alpha)$ denote the single-split RSPIM
interval for coordinate $j$ constructed from the support $\widehat S$ and split
$(I_{\mathrm{sel}},I_{\mathrm{inf}})$.

Under Assumptions~\textup{(L1)}–\textup{(L2)}, for any coordinate $j$ and any level $1-\alpha\in(0,1)$,
\[
\sup_{\mathcal{S}\in\mathcal{S}}\;\sup_{\beta_0}
\mathbb{P}^{(\mathcal{S})}_{\beta_0}\{\beta_{0j}\notin C^{(\mathcal{S})}_{j,n}(\alpha)\} \le \alpha.
\]
\end{proposition}

\subsection{High-dimensional sparse regime}\label{sec:theory-hd}

We now move from the selector-agnostic finite-sample guarantees to a sparse high-dimensional regime. We embed the single-split construction in a sparse linear model with $p=p_n$ and $S_0 = \{j:\,\beta_{0j}\neq 0\}$ of cardinality $s_n = |S_0|$, and define
\[
\mathcal{B}_{s_n} = \{\beta\in\mathbb{R}^{p_n}:\,\|\beta\|_0\le s_n\}.
\]

In contrast to the finite-sample results, which are selector-agnostic, the high-dimensional analysis requires additional structure on the selector: we assume that the sequence of split-based selectors $\{\mathcal{S}_n\}$ satisfies a screening and support-size condition~\textup{(HD4)}. The full statements of Assumptions~\textup{(HD1)}–\textup{(HD4)} and design/tuning Conditions~\textup{(C1)}–\textup{(C3)} are collected in Appendix~\ref{app:HD}. In words, they impose (i)~a sparse linear model with $s_n \log p_n / n \to 0$, (ii)~sub-Gaussian design regularity and uniform well-conditioning on $O(s_n)$-sparse sets, (iii)~balanced split proportions, and (iv)~a screening property: with high probability $\widehat S_n$ contains $S_0$ and has cardinality of order $s_n$.

\begin{proposition}[Lasso stability selection implies (HD4)]\label{prop:hd4-lasso}
Under Assumptions~\textup{(HD1)}–\textup{(HD3)} and Conditions~\textup{(C1)}–\textup{(C3)}, let $\widehat S_n$ denote the set returned by lasso stability selection on $I_{\mathrm{sel}}$. Then there exists a constant $C<\infty$ such that,
uniformly over $\beta_0\in\mathcal{B}_{s_n}$,
\[
\mathbb{P}_{\beta_0}\bigl(S_0\subseteq\widehat S_n,\;|\widehat S_n|\le C s_n\bigr)\to 1.
\]
In particular, Assumption~\textup{(HD4)} holds for this selector.
\end{proposition}

\begin{theorem}[Asymptotic strong validity over sparse classes]\label{thm:hd-strong}
Suppose Assumptions~\textup{(L1)}–\textup{(L2)} and \textup{(HD1)}–\textup{(HD4)} hold. Fix an index
$j\in\{1,\dots,p_n\}$ and consider $\beta_{0j}$ under sparse parameter values
$\beta_0\in\mathcal{B}_{s_n}$ with $j\in S_0(\beta_0)$. For each $n$, let $C_{j,n}(\alpha)$ denote
the single-split RSPIM interval for $\beta_{0j}$ at level $1-\alpha\in(0,1)$, defined whenever
$j\in\widehat S_n$. Then
\[
\limsup_{n\to\infty}\;\sup_{\beta_0\in\mathcal{B}_{s_n}: j\in S_0(\beta_0)}
\mathbb{P}_{\beta_0}\{\beta_{0j}\notin C_{j,n}(\alpha)\mid j\in\widehat S_n\} \le \alpha,
\]
and, equivalently, for the plausibility contour $\pi_{n,j}(\cdot)$,
\[
\limsup_{n\to\infty}\;\sup_{\beta_0\in\mathcal{B}_{s_n}: j\in S_0(\beta_0)}
\mathbb{P}_{\beta_0}\{\pi_{n,j}(\beta_{0j})\le u \mid j\in\widehat S_n\} \le u,\qquad u\in[0,1].
\]
\end{theorem}

\begin{corollary}[Unconditional strong validity in the sparse regime]
\label{cor:uncond-strong}
Under the conditions of Theorem~\ref{thm:hd-strong}, uniformly over $\beta_0\in\mathcal{B}_{s_n}$ with $j\in S_0(\beta_0)$,
\[
\limsup_{n\to\infty}
\sup_{\beta_0\in\mathcal{B}_{s_n}: j\in S_0(\beta_0)}
P_{\beta_0}\bigl\{\pi_{n,j}(\beta_{0j}) \le u\bigr\}
\;\le\; u, \qquad u\in[0,1],
\]
and the unconditional non-coverage probability satisfies
\[
\limsup_{n\to\infty}
\sup_{\beta_0\in\mathcal{B}_{s_n}: j\in S_0(\beta_0)}
P_{\beta_0}\bigl\{\beta_{0j}\notin C_{j,n}(\alpha)\bigr\}
\;\le\; \alpha.
\]
\end{corollary}

\begin{corollary}[Oracle-equivalent efficiency under selection consistency]\label{cor:oracle}
Under the conditions of Theorem~\ref{thm:hd-strong} and the additional selection consistency
assumption $\mathbb{P}(\widehat S_n=S_0)\to 1$, let $C^{\mathrm{oracle}}_{j,n}(\alpha)$ denote the RSPIM
interval obtained by running the inference step with $S=S_0$ known. Then, for any fixed
coordinate $j$ and any level $1-\alpha\in(0,1)$,
\[
\frac{\operatorname{diam}\bigl(C_{j,n}(\alpha)\bigr)}
{\operatorname{diam}\bigl(C^{\mathrm{oracle}}_{j,n}(\alpha)\bigr)} \;\xrightarrow{P}\; 1.
\]
Moreover, both intervals have length of order $O_P(n_{\mathrm{inf}}^{-1/2})$.
\end{corollary}

\subsection{Orthogonalized and robust extension}\label{sec:orthogonalized}

We next consider inference on a low-dimensional target in the presence of high-dimensional nuisance components.
Our starting point is the orthogonal-score construction with cross-fitting, formulated in a way tailored to RSPIM.
We construct a Neyman-orthogonal score $\psi(W;\theta,\eta)$, estimate nuisance components via flexible machine learning or high-dimensional regression, and then form a studentized score whose distribution is asymptotically pivotal under suitable rate and regularity conditions.

Let $\psi(W;\theta,\eta)$ denote a Neyman-orthogonal score for $\theta$, where $W$ collects the observed variables. The orthogonal RSPIM construction proceeds by cross-fitting nuisance estimators, forming a studentized score statistic
\begin{equation}\label{eq:orth-score}
T_n(\theta) = \frac{\sqrt{n_{\mathrm{inf}}}\,\widehat U_n(\theta)}{\widehat\sigma(\theta)},
\end{equation}
where $\widehat U_n(\theta)$ is the cross-fitted score and $\widehat\sigma(\theta)$ is a consistent estimator of the standard deviation of $\sqrt{n_{\mathrm{inf}}}\widehat U_n(\theta)$ conditional on $\mathcal{G}$. The PIM validification recipe~\eqref{eq:plaus-contour} is then applied to $T_n(\theta)$ to obtain a consonant plausibility contour $\pi_n^{\mathrm{orth}}(\theta)$.

Under Assumptions~\textup{(O1)}–\textup{(O4)} (stated in Appendix~\ref{app:theory-proofs}), which encode Neyman orthogonality, nuisance convergence at rate $o_P(n_{\mathrm{inf}}^{-1/4})$, an asymptotic linearity and CLT, and a uniform bootstrap approximation, the following holds.

\begin{theorem}[Orthogonalized and robust RSPIM]
\label{thm:orth-RSPIM}
Suppose Assumptions {\rm(O1)}–{\rm(O4)} hold and the validification step is carried out with the
studentized orthogonal score $T_n(\theta)$ in~\eqref{eq:orth-score}.
Then the resulting
plausibility contour $\pi_n^{\mathrm{orth}}(\theta)$ satisfies
\[
\limsup_{n\to\infty}
\sup_{\theta_0} P_{\theta_0}\bigl\{ \pi_n^{\mathrm{orth}}(\theta_0) \le u \bigr\}
\;\le\; u, \qquad u \in [0,1],
\]
so the associated consonant sets deliver asymptotically strongly valid inference for~$\theta_0$.
If heteroskedastic or heavy-tailed errors are present but Assumption {\rm(O4)} holds for a wild
bootstrap, the same conclusion remains valid.
\end{theorem}

\paragraph{Example: partially linear regression.}
Consider i.i.d.\ observations $Z_i = (Y_i,D_i,X_i)$ from the partially linear model
\[
Y = D\,\theta_0 + g_0(X) + \varepsilon,\qquad \mathbb{E}(\varepsilon\mid D,X)=0,
\]
where $D$ is a scalar treatment, $X\in\mathbb{R}^{p_n}$ is high-dimensional, and $\theta_0$ is the scalar target of interest.
We use the standard Neyman-orthogonal score
\[
\psi(Z;\theta,\eta)
= \{D-m(X)\}\{Y-g(X)-\theta\{D-m(X)\}\},\qquad \eta = (g,m),
\]
where $m_0(x) = \mathbb{E}(D\mid X=x)$ and $\eta_0 = (g_0,m_0)$. Under Conditions~\textup{(PL1)}–\textup{(PL3)} (stated in Appendix~\ref{app:theory-proofs}), which impose moment bounds, approximate sparsity with $s_n^2(\log p_n)^2/n_{\mathrm{inf}}\to 0$, and restricted eigenvalue conditions, we obtain the following.

\begin{corollary}[Orthogonal RSPIM in the partially linear model]\label{cor:pl-orth}
Suppose the data follow the partially linear model above and that
\textup{(PL1)}–\textup{(PL3)} hold.
Construct the orthogonal, cross-fitted
score and its wild-bootstrap calibration as described in
Section~\ref{sec:orthogonalized}.
Then the resulting plausibility contour
$\pi_{n,\mathrm{orth}}(\theta)$ is asymptotically strongly valid for $\theta_0$:
\[
  \limsup_{n\to\infty}\;\sup_{\theta_0}
  \mathbb{P}_{\theta_0}\bigl\{\pi_{n,\mathrm{orth}}(\theta_0)\le u\bigr\} \le u,
  \qquad u\in[0,1].
\]
\end{corollary}

\subsection{Maxitive multi-split aggregation}\label{sec:multi}

Finally, we analyze the maxitive multi-split aggregation introduced in Section~\ref{subsec:multisplit}. Let
$C^{(r)}_{j,n}(\alpha)$ denote the consonant interval obtained from split $r$ at level $1-\alpha$,
and let $\pi_n^{(r)}$ be the corresponding plausibility contour for $\beta_j$. Define the maxitive
aggregate contour and union-of-splits interval by
\[
\pi^{\max}_{n,j}(\beta_j) = \max_{1\le r\le R} \pi^{(r)}_{n,j}(\beta_j),\qquad
C^{\max}_{j,n}(\alpha)
= \{\beta_j:\,\pi^{\max}_{n,j}(\beta_j)>\alpha\}
= \bigcup_{r=1}^R C^{(r)}_{j,n}(\alpha).
\]

\begin{proposition}[Maxitive aggregation preserves strong validity]\label{prop:maxitive}
If each split-wise plausibility contour $\pi_n^{(r)}$ is strongly valid conditional on its split-specific
$\sigma$-field $\mathcal{G}^{(r)}$, then for every coordinate $j$, every level $1-\alpha\in(0,1)$, and
every $u\in[0,1]$,
\[
\sup_{\theta_0}\mathbb{P}_{\theta_0}\{\pi^{\max}_{n,j}(\theta_0)\le u\} \le u
\quad\text{and}\quad
\mathbb{P}_{\theta_0}\{\beta_{0j}\in C^{\max}_{j,n}(\alpha)\} \ge 1-\alpha.
\]
In particular, $C^{\max}_{j,n}(\alpha)$ is the $(1-\alpha)$ upper-level set of a strongly valid contour.
\end{proposition}

For completeness, we also record the intersection-of-splits interval
$C^{\cap}_{j,n}(\alpha) = \bigcap_{r=1}^R C^{(r)}_{j,n}(\alpha)$,
which can be viewed as an aggressive efficiency summary and may undercover. We report
$C^{\cap}_{j,n}(\alpha)$ only as a diagnostic in the simulation study; it is not used as a primary
inferential device.

\section{Computation}\label{sec:computation}

\subsection{Single-split implementation}

All refitting and validification steps are performed on the inference subsample $I_{\mathrm{inf}}$. On the selected support $S = \widehat S$ we compute the least-squares estimate
\[
\widehat\beta_S = (X_S^\top X_S)^{-1} X_S^\top Y,
\qquad
\widehat\sigma^2 
= \frac{1}{n_{\mathrm{inf}} - |S|}\bigl\|Y - X_S \widehat\beta_S\bigr\|_2^2,
\]
where $X_S = X_{I_{\mathrm{inf}},S}$ and $Y = Y_{I_{\mathrm{inf}}}$. Writing $V_S = (X_S^\top X_S)^{-1}$ and $v_{jj}$ for its $(j,j)$ entry, the basic RSPIM construction uses closed-form $t/F$ quantiles to obtain consonant plausibility sets and their coordinate-wise projections. For likelihood-ratio-based contours, we also report the Wilks-style approximation
\[
\pi_z(\theta) \approx 1 - F_{|S|}\{-2\log R(z,\theta)\},
\]
where $R(z,\theta)$ is the profile likelihood ratio and $F_{|S|}$ is the $\chi^2_{|S|}$ distribution function.

\subsection{Multi-split aggregation}

Independent splits are trivially parallelizable. For each split we cache $(X_S^\top X_S)^{-1}$ and store the per-split intervals, which are then combined by union or intersection according to the maxitive aggregation rule in Section~\ref{sec:multi}. The cost of this aggregation step is negligible relative to selection and refitting; in practice, computation is dominated by the initial lasso or other high-dimensional selector.

\subsection{Carving and bootstrap extensions}

If optional data carving is used, we allow part of the selection sample to be reused for refitting by enlarging $I_{\mathrm{inf}}$. This breaks the exact independence between selection and inference, so the finite-sample pivotal arguments from Section~\ref{sec:finite-sample} no longer apply. We therefore label carved variants as approximate and treat their guarantees as asymptotic.

For heteroskedastic or heavy-tailed errors, we replace the exact $t/F$ pivot by a wild-bootstrap approximation on $I_{\mathrm{inf}}$: residuals are multiplied by i.i.d.\ multipliers (e.g., Mammen's two-point or Rademacher weights), and the corresponding $t$- or likelihood-ratio statistics are recomputed. The validification step then uses the empirical bootstrap distribution in place of the parametric reference distribution; only the quantile entering the contour construction changes. The orthogonalized extension in Section~\ref{sec:orthogonalized} is implemented in the same way, with the score-based statistic substituted for the refitted $t/F$ pivot.

\subsection{Calibration and default settings}

The theoretical guarantees in Section~\ref{sec:theory} are stated for the \emph{raw} RSPIM intervals, with no rescaling of the pivotal statistics; equivalently, they correspond to the case $c=1$ in the notation of Section~\ref{sec:experiments}. In the simulation study, when comparing RSPIM to other post-selection procedures, we sometimes introduce a scalar ``shrink factor'' $c>0$ that rescales the per-coordinate pivot or its critical value, and choose $c$ so that the empirical conditional coverage across selected coordinates is close to the target level $1-\alpha$. This calibration is applied symmetrically to all methods and is purely an engineering device for reporting efficiency at equal conditional coverage. The finite-sample and asymptotic strong-validity guarantees in Section~\ref{sec:theory} always refer to the uncalibrated RSPIM with $c=1$ and remain valid regardless of whether such empirical tuning is used in a given application.

\section{Experiments}\label{sec:experiments}

We organize the empirical study by first assessing \emph{calibration}, then comparing \emph{efficiency} at equal conditional coverage, and finally probing \emph{robustness} to non-Gaussian noise. Calibration is evaluated using three standard diagnostics: (i) histograms and QQ plots of null plausibility/$p$-values, (ii) coverage-versus-nominal level curves for plausibility intervals, and (iii) a ``false-confidence'' check that records the frequency with which plausibility exceeds $1-\alpha$ on selected null coordinates; precise definitions are given in Appendix~\ref{app:calibration-diagnostics}. Unless otherwise noted, the default selector is lasso with stability selection; de-biased lasso and exact selective inference are used as baselines. We consider high-dimensional linear models with Gaussian features, varying correlation $\rho\in\{0,0.5\}$, sparsity $s\in\{5,10\}$, sample sizes $n\in\{100,200\}$, and dimensions $p\in\{500,2000\}$; robustness is assessed under heteroskedastic and heavy-tailed noise. Detailed experimental designs, calibration diagnostics, and efficiency comparison methodology are summarized in Appendix~\ref{app:experiments}.

\subsection{Module A: Baseline calibration under no post-selection}
Under Gaussian noise without post-selection, the empirical null plausibility (equivalently, $p$-values under our pivot) closely matches $\mathrm{Unif}(0,1)$, with QQ plots showing only minor conservative deviations. This supports the finite-sample calibration of our pivot prior to any selection effects.
\begin{figure}[htbp]
  \centering
  \includegraphics[width=.95\linewidth]{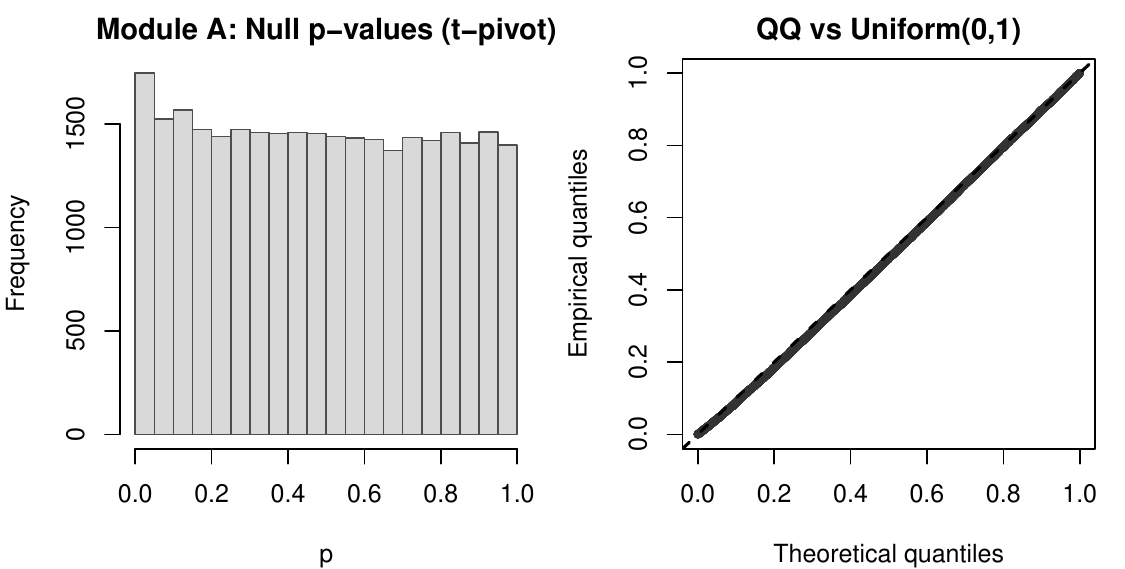}
  \caption{Module A: Null plausibility/$p$-value diagnostics under the Gaussian design. Histograms and QQ-plots against $\mathrm{Unif}(0,1)$ demonstrate strong validity; slight conservatism in the extreme tails is benign for IM/PIM.}
  \label{fig:A-pvals}
\end{figure}

\subsection{Module B: Gaussian high-dimensional setting}
Figure~\ref{fig:B-contour} displays representative plausibility contours for a selected coefficient under single-split and union-of-splits constructions. The single-split contour is sharply peaked around the refitted estimate, whereas the maxitive union produces a flatter contour that still assigns substantial plausibility in a neighborhood of the true coefficient. This provides a genuinely possibilistic diagnostic: maxitive aggregation trades local sharpness for ``strong but diffuse'' evidence, revealing when repeated splits lead to diffuse post-selection uncertainty. At equal conditional coverage ($1-\alpha=0.90$), the single-split achieves shorter intervals with comparable power (Figures~\ref{fig:B-calib}--\ref{fig:B-milpower}), making it our default for high-dimensional inference.

\begin{figure}[htbp]
  \centering
  \includegraphics[width=.7\linewidth]{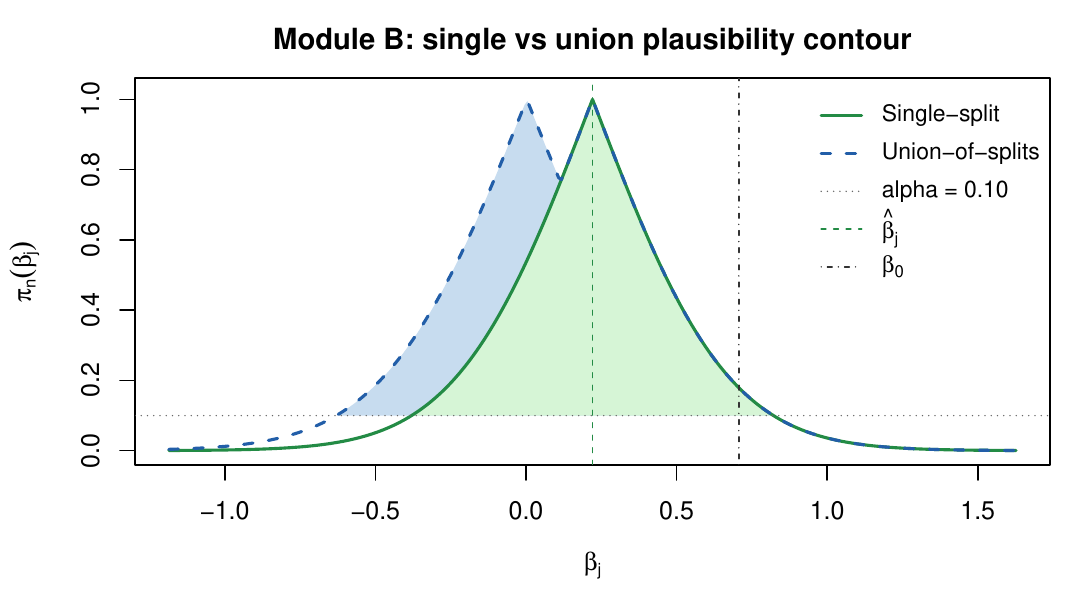}
  \caption{Module B: single-split versus union-of-splits plausibility contour for a representative selected coefficient in the high-dimensional Gaussian design. The single-split contour is more sharply peaked around the refitted estimate, while the maxitive union yields a flatter but still high-plausibility region that continues to cover the true coefficient. Design: $n=100$, $p=500$, $s=10$, $\rho=0.5$, $\mathrm{SNR}=5$.}
  \label{fig:B-contour}
\end{figure}

\begin{figure}[htbp]
  \centering
  % 上面一行：coverage-vs-nominal 曲线
  \includegraphics[width=.95\linewidth]{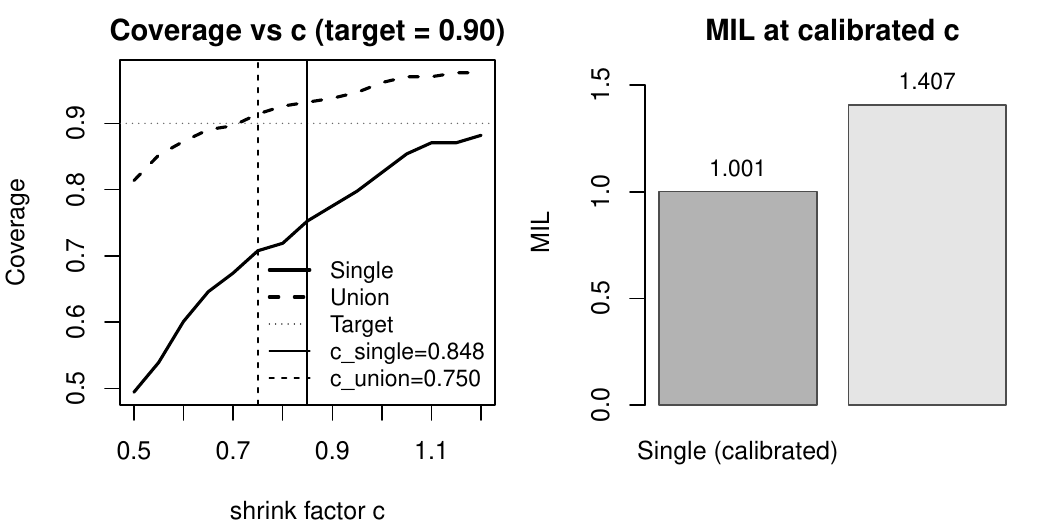}\\[0.5em]
  % 下面一行：对应的 c^\star 标度
  \includegraphics[width=.95\linewidth]{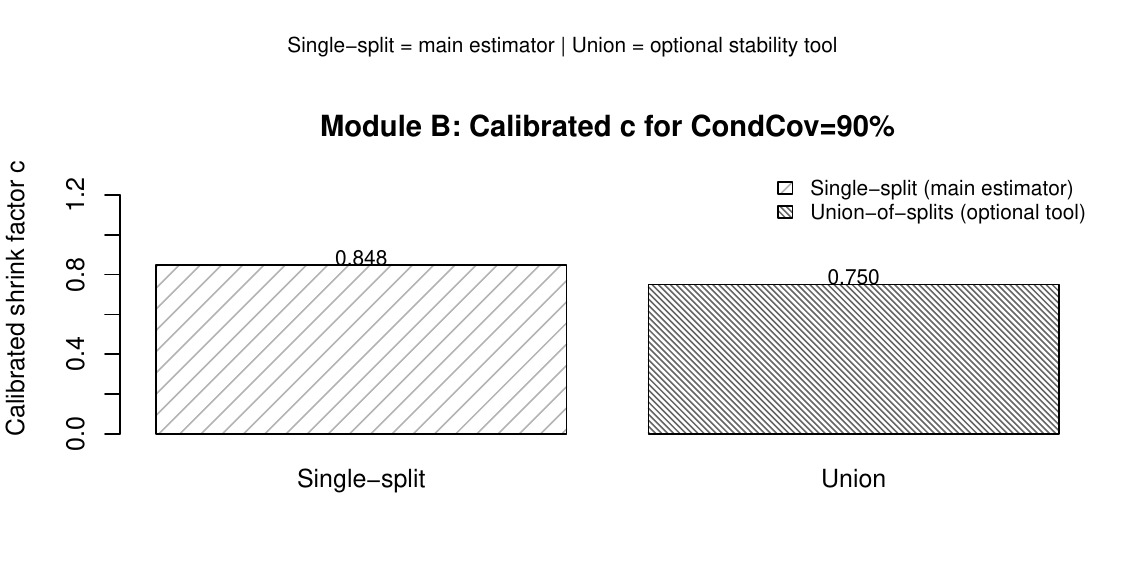}
  \caption{Module B: (Top) empirical coverage of $(1-\alpha)$ plausibility intervals versus nominal level, for a grid of calibration factors $c$, under the high-dimensional Gaussian design with stability-selected lasso. (Bottom) selected calibration factors $c_m^\star$ at target conditional coverage $1-\alpha=0.90$ across methods. Shaded ribbons (when present) indicate Monte Carlo variability. The union-of-splits aggregator is typically slightly more conservative than the single-split procedure.}
  \label{fig:B-calib}
\end{figure}

\begin{figure}[htbp]
  \centering
  \includegraphics[width=.95\linewidth]{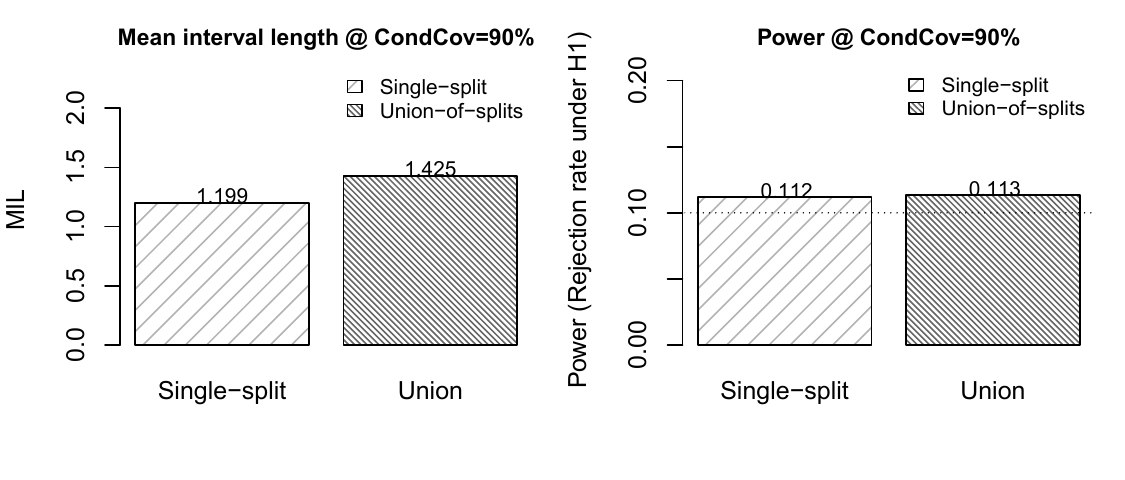}
  \caption{Module B: efficiency at equal conditional coverage ($1-\alpha=0.90$). For each method we compare median interval length (MIL) and empirical power, conditional on selection, using the calibrated $c_m^\star$ from Figure~\ref{fig:B-calib}. The single-split RSPIM intervals are uniformly shorter than, or comparable to, those from the union-of-splits aggregator, with essentially no loss in power.}
  \label{fig:B-milpower}
\end{figure}

\subsection{Module C: Robustness to heteroskedasticity and heavy tails}
We assess robustness under heteroskedastic and heavy-tailed noise using wild-bootstrap $t$ pivots. Figure~\ref{fig:C-pvals} shows null plausibility/$p$-values remain near-uniform with slight conservatism, confirming robust calibration. After calibration to $1-\alpha=0.90$, the single-split maintains conditional coverage with modest increases in interval length (Figures~\ref{fig:C-calib}--\ref{fig:C-milpower}). The union-of-splits may produce no effective intervals in stress configurations and thus serves as a conservative stability check rather than the primary estimator.
\begin{figure}[htbp]
  \centering
  \includegraphics[width=.95\linewidth]{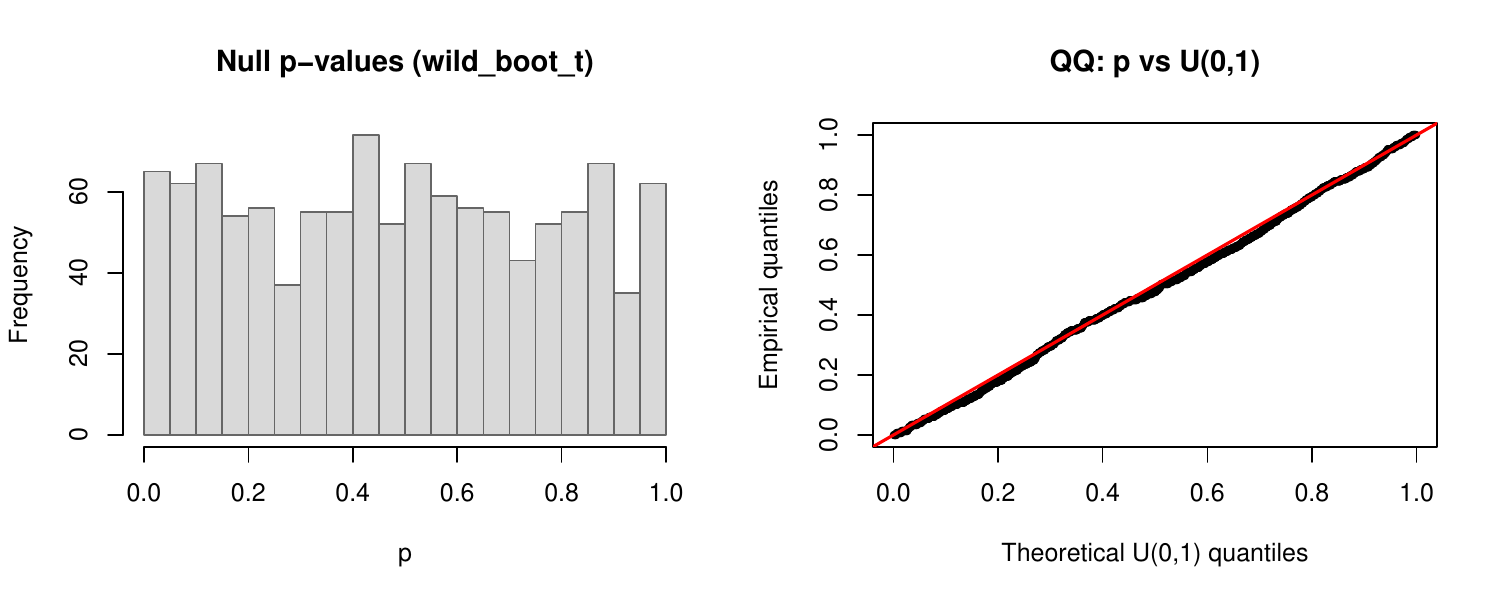}
  \caption{Module C (wild bootstrap): Null plausibility/$p$-value diagnostics under heteroskedastic noise. Near-uniform histograms and QQ plots track $\mathrm{Unif}(0,1)$ with slight conservatism, supporting robust calibration.}
  \label{fig:C-pvals}
\end{figure}

\begin{figure}[htbp]
  \centering
  \includegraphics[width=.95\linewidth]{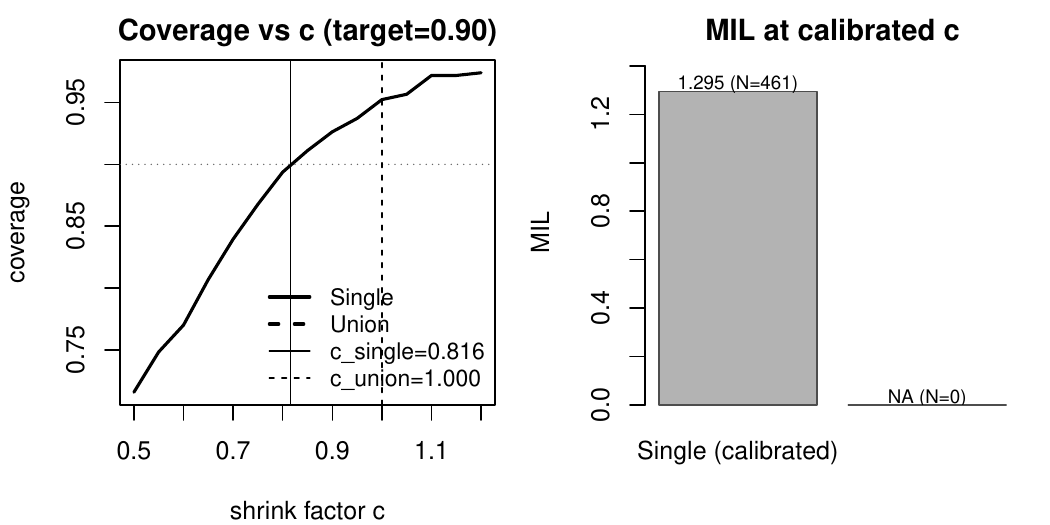}
  \caption{Module C (wild bootstrap): Coverage-versus-$c$ under robust noise. Calibrated pivots achieve nominal coverage across levels; the union-of-splits curve is more conservative.}
  \label{fig:C-calib}
\end{figure}

\begin{figure}[htbp]
  \centering
  \includegraphics[width=.95\linewidth]{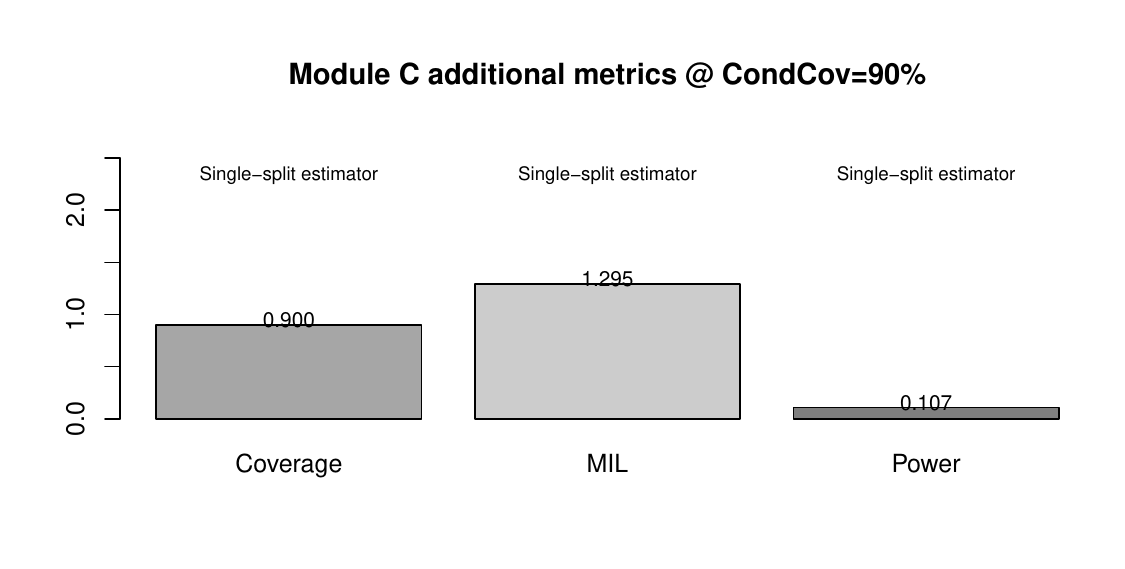}
  \caption{Module C (wild bootstrap): Efficiency at equal conditional coverage (calibrated $c$). Left: median interval length (MIL). Right: null rejection rate at nominal~$\alpha$. Union-of-splits results are omitted when $N=0$.}
  \label{fig:C-milpower}
\end{figure}

\subsection{Module D: Orthogonalized inference vs.\ de-biased lasso}
We compare orthogonalized, cross-fitted wild-bootstrap pivots with de-biased lasso. Figure~\ref{fig:D-null} confirms strong validity after calibration. At equal conditional coverage ($1-\alpha=0.90$), orthogonalized pivots typically require $c\approx 1$ and exhibit stable calibration under heteroskedastic/heavy-tailed noise, while de-biased lasso achieves the target after larger rescaling, yielding more aggressive intervals with increased sensitivity (Figures~\ref{fig:D-calib}--\ref{fig:D-milpower}, Table~\ref{tab:D_condcov}).

\begin{figure}[htbp]
  \centering
  \includegraphics[width=.95\linewidth]{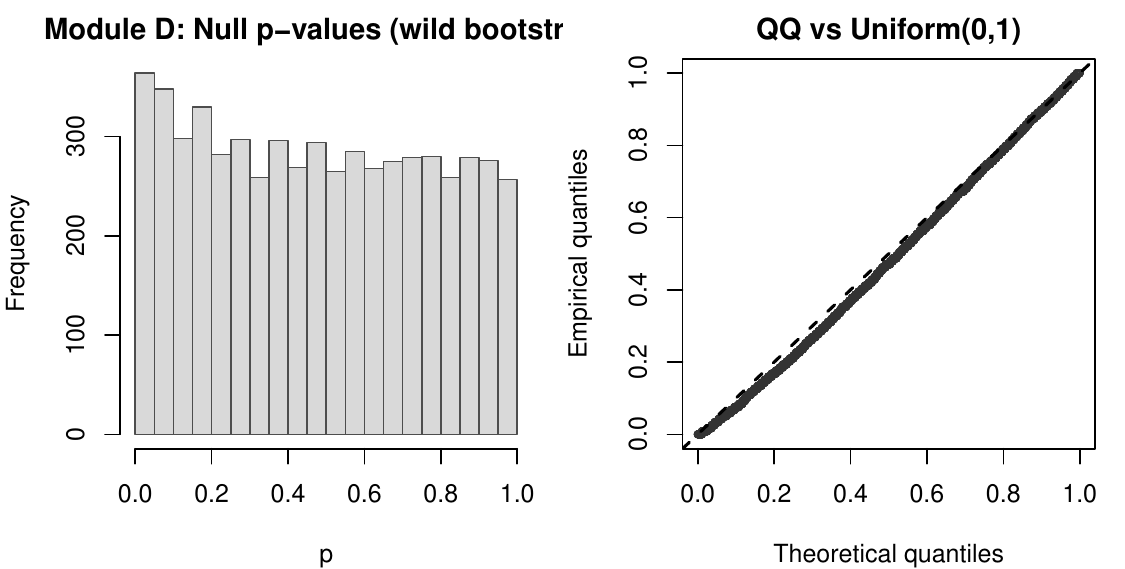}
  \caption{Module D: Null plausibility/$p$-value diagnostics for the orthogonalized, cross-fitted construction. Histograms and QQ-plots against $\mathrm{Unif}(0,1)$ confirm strong validity after calibration.}
  \label{fig:D-null}
\end{figure}

\begin{figure}[htbp]
  \centering
  \includegraphics[width=.95\linewidth]{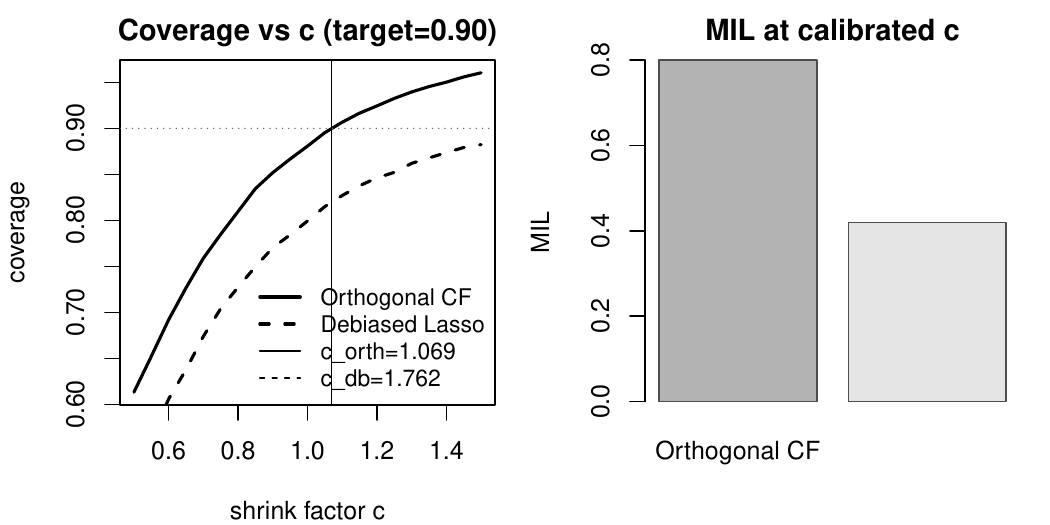}
  \caption{Module D: Coverage-versus-$c$ for orthogonal, cross-fitted pivots versus de-biased lasso. After calibration, both attain nominal coverage; orthogonal pivots typically operate with $c\approx 1$.}
  \label{fig:D-calib}
\end{figure}

\begin{figure}[htbp]
  \centering
  \includegraphics[width=.95\linewidth]{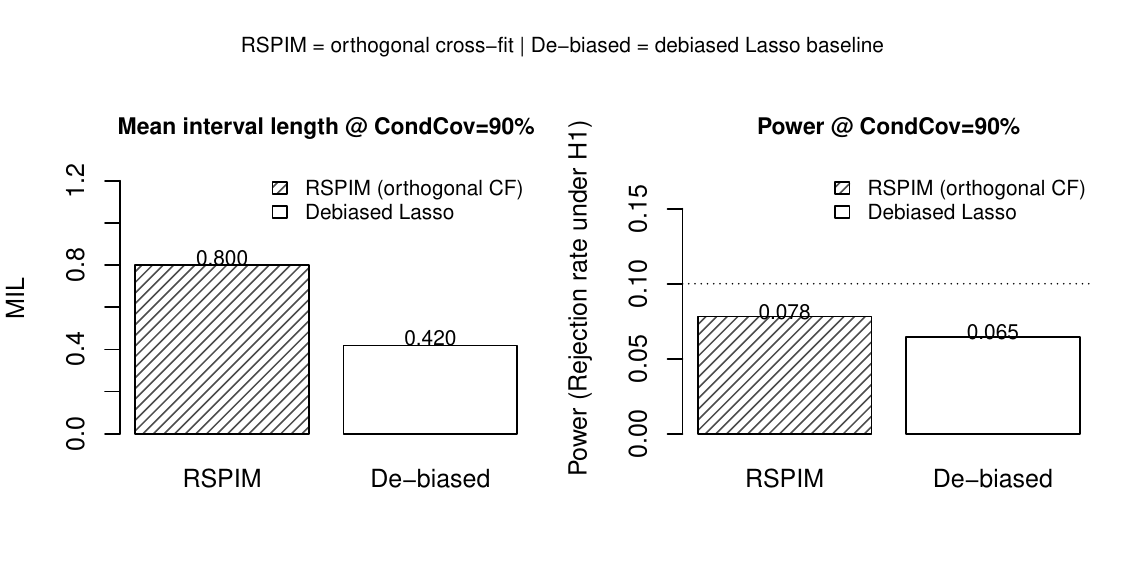}
  \caption{Module D: Efficiency at equal conditional coverage (calibrated $c$). Left: median interval length (MIL), relative to de-biased lasso. Right: null rejection at nominal~$\alpha$ (reference line shown).}
  \label{fig:D-milpower}
\end{figure}

% Auto-generated: Module D table (CondCov)
\begin{table}[htbp]
\centering
\begin{tabular}{lcccc}
\hline
Method & Coverage & MIL & Power & $c$ \\ \hline
Orthogonal CF & 0.9010 & 0.800 & 0.0781 & 1.069 \\
Debiased Lasso & 0.9008 & 0.420 & 0.0646 & 1.762 \\ \hline
\end{tabular}
\caption{Equal-conditional-coverage comparison (Module D). Intervals are formed from cross-fitted orthogonal scores (ours) and from debiased Lasso; both are calibrated to the same conditional coverage before comparing MIL/Power.}
\label{tab:D_condcov}
\end{table}

\subsection{Module E: Comparison with polyhedral exact selective inference}
We benchmark RSPIM against polyhedral selective inference \citep{Lee2016} in a challenging regime: $(n,p,\rho)=(100,200,0.9)$ with weak signals ($\beta_{0j}\approx 0.4$). Figure~\ref{fig:E-poly-stress} shows that polyhedral intervals achieve shorter medians but exhibit heavy-tailed distributions with a non-negligible fraction of infinite-length intervals (red asterisks) due to near-singular conditioning. RSPIM pays a fixed efficiency tax from sample splitting but delivers uniformly stable, finite uncertainty quantification, providing a robust fail-safe alternative when exact conditional inference becomes numerically unstable.

\begin{figure}[htbp]
  \centering
  \includegraphics[width=.95\linewidth]{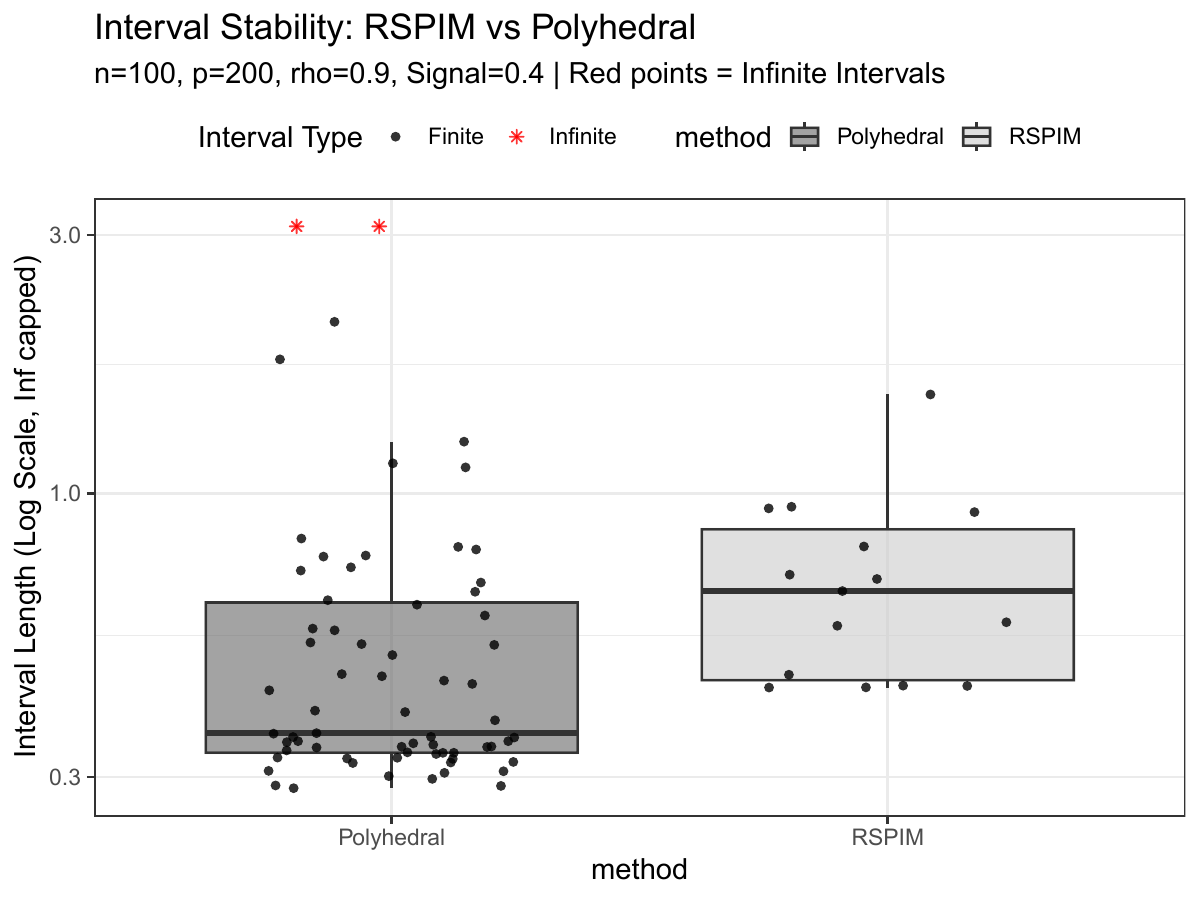}
  \caption{Module E: Interval stability under stress $(n=100,p=200,\rho=0.9)$ with weak
  signals of size $0.4$. Boxplots show the log interval lengths for the polyhedral exact
  selective-inference method and for RSPIM; red asterisks mark infinite-length intervals
  produced by the polyhedral method. While polyhedral intervals can be shorter in median,
  they exhibit catastrophic instability and high variance, whereas RSPIM yields uniformly
  finite, stable intervals.}
  \label{fig:E-poly-stress}
\end{figure}

\subsection{Real-data illustration: riboflavin gene expression}\label{sec:real-data}
We illustrate RSPIM on the riboflavin data ($n=71$, $p=4088$) from the \texttt{hdi} package. Using lasso stability selection with $R=50$ random splits, the three most frequently selected genes are LYSC\_at, XLYA\_at, and YOAB\_at (selection frequencies $0.19$--$0.30$). Table~\ref{tab:riboflavin-intervals} compares union-of-splits RSPIM intervals ($1-\alpha=0.90$, uncalibrated $c=1$) with de-biased lasso intervals. RSPIM intervals are wider and typically contain zero, reflecting deliberate conservatism after maxitive aggregation, while de-biased lasso intervals exclude zero for all three genes. Figure~\ref{fig:real-plaus} shows the plausibility contour for LYSC\_at is relatively flat around zero, illustrating that RSPIM regards both near-zero and moderately negative effects as highly plausible, highlighting post-selection uncertainty at this sample size.

\begin{table}[htbp]
  \centering
  \caption{Riboflavin data: union-of-splits RSPIM and de-biased lasso $90\%$ intervals for the
  three most frequently selected genes (interval endpoints rounded to two decimal places).}
  \label{tab:riboflavin-intervals}

  \resizebox{\linewidth}{!}{%
    \begin{tabular}{lccc}
      \toprule
      Gene & Selection frequency & RSPIM union-of-splits $90\%$ interval & De-biased lasso $90\%$ interval \\
      \midrule
      LYSC\_at & $0.30$ & $[-0.81,\; 0.07]$ & $[-0.39,\; -0.06]$ \\
      XLYA\_at & $0.30$ & $[0.03,\; 0.91]$  & $[0.17,\; 0.49]$ \\
      YOAB\_at & $0.19$ & $[-0.78,\; 0.10]$ & $[-0.44,\; -0.10]$ \\
      \bottomrule
    \end{tabular}%
  }
\end{table}

\begin{figure}[htbp]
  \centering
  \includegraphics[width=.7\linewidth]{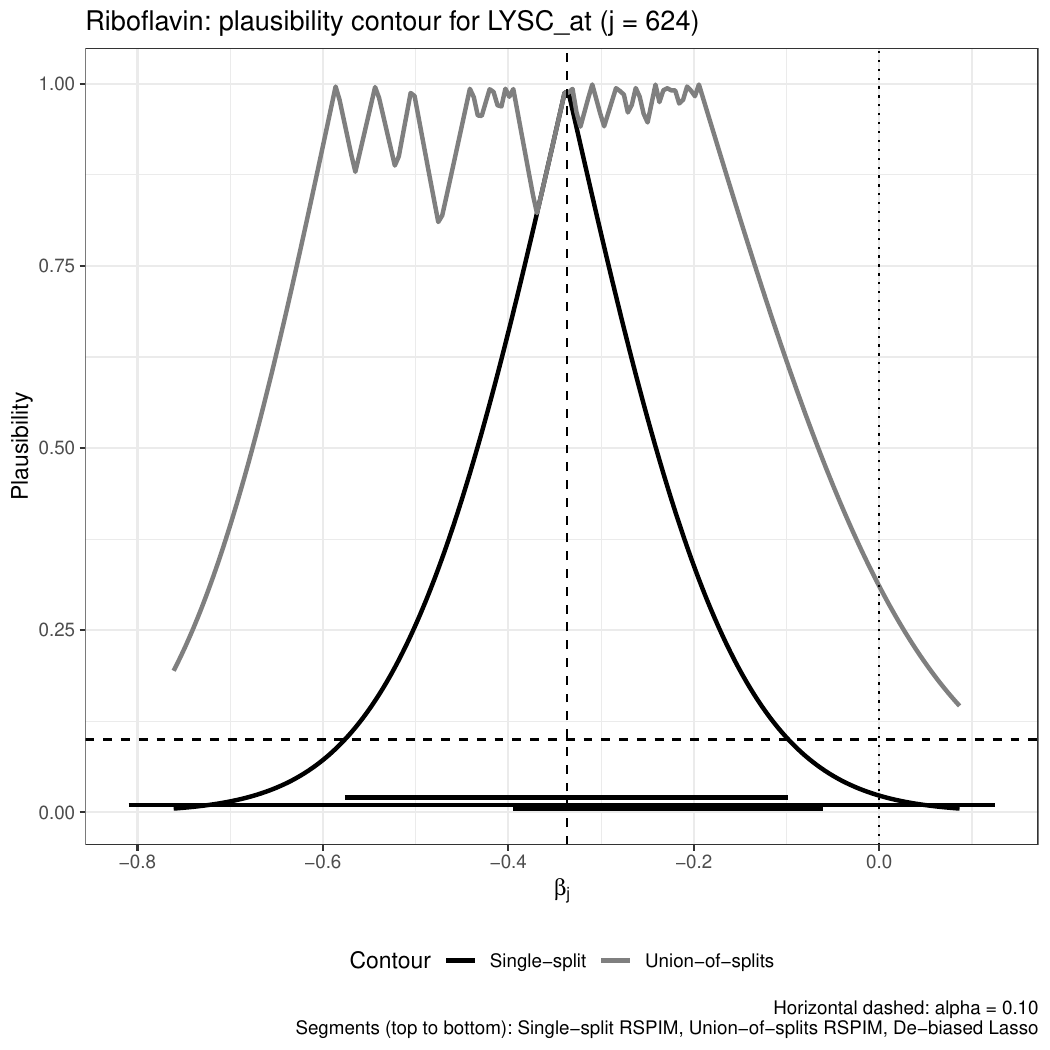}
  \caption{Real-data illustration (riboflavin). Plausibility contour for the coefficient of gene
  LYSC\_at under the union-of-splits RSPIM construction (solid line), together with the corresponding
  $90\%$ RSPIM interval (horizontal bar) and the $90\%$ de-biased lasso interval (dashed bar).}
  \label{fig:real-plaus}
\end{figure}

\subsection{Summary of empirical findings}
The empirical study supports three main messages: (i) validified pivots are well calibrated or mildly conservative under both Gaussian and non-Gaussian noise, with null plausibility/$p$-values near $\mathrm{Unif}(0,1)$ and coverage tracking nominal levels; (ii) at equal conditional coverage, single-split RSPIM intervals are competitive with state-of-the-art methods (de-biased lasso, polyhedral selective inference), while maxitive aggregation yields a conservative but stable union-of-splits summary; and (iii) plausibility contours provide transparent diagnostics of how strongly and stably the data support post-selection effect sizes, revealing when apparently significant signals rest on fragile selection events. Additional results with alternative selectors and detailed numerical comparisons are provided in Appendix~\ref{app:experiments}.

\section{Discussion}\label{sec:discussion}

\subsection{Sample splitting and efficiency}
Sample splitting inevitably trades some efficiency for selector-agnostic finite-sample guarantees. In our construction this trade-off is controlled and transparent: the refitted model on the inference subsample is low-dimensional and admits exact $t/F$ pivots, while the selector is treated as a black box on the selection subsample. The resulting intervals are typically longer than those of procedures that reuse the full sample without splitting, but they remain numerically stable and plausibility-calibrated across designs, including highly correlated or weak-signal regimes where conditioning-based methods can fail.

\subsection{Novelty and scope}
Methodologically, our main contribution is to organize a standard split-and-refit scheme within the possibilistic IM/PIM framework. This leads to selector-agnostic plausibility contours that inherit classical linear-model structure but are equipped with strong validity guarantees and visual diagnostics. Conceptually, RSPIM shows how high-dimensional regularization can be combined with IM/PIM ideas in a principled way: regularization is used solely for screening on one subsample, while validification is carried out on an independent subsample where the model is effectively low-dimensional.

\subsection{Orthogonalized and robust extensions}
For problems with low-dimensional targets and high-dimensional nuisance components, we augment the basic RSPIM template with Neyman-orthogonal scores and cross-fitting. Under high-level orthogonality and rate conditions, these extensions yield asymptotically strongly valid plausibility contours even under heteroskedastic or heavy-tailed errors. The partially linear example illustrates how these ideas can be implemented in semiparametric models with high-dimensional controls, and how wild-bootstrap validification can be combined with orthogonalization to preserve robustness while retaining the interpretation of plausibility contours as diagnostic tools.

\subsection{Choice of default selector}
In the empirical study we take lasso combined with stability selection as the default selector, balancing power and screening reliability in typical sparse designs. The theoretical development, however, is deliberately modular: the finite-sample guarantees in Gaussian linear models hold uniformly over all selectors that use only the selection data, and the high-dimensional results apply to any selector satisfying the screening property (HD4). Additional selectors, including knockoffs and forward stepwise, are considered in the supplement to illustrate this modularity.

\subsection{Limitations and future directions}
The strongest finite-sample guarantees are obtained for post-selection coordinates in Gaussian linear models with homoskedastic errors. Extensions to full-model coordinates and to non-Gaussian or heteroskedastic settings rely on asymptotic arguments or bootstrap approximations and are therefore conservative rather than exact. Maxitive intersection aggregation can undercover and is used only as an efficiency diagnostic, while multi-split procedures incur additional computational cost in exchange for stability information. It would be of interest to investigate adaptive splitting schemes, refinements under weaker design assumptions, and extensions beyond linear models, for example to generalized linear, time-to-event, or nonparametric models.

\subsection{Concluding remarks}
RSPIM provides a modular possibilistic framework for high-dimensional post-selection inference that separates regularization from validification through sample splitting. The resulting plausibility contours retain familiar $t/F$-based interval forms while supplying strong validity guarantees and graphical diagnostics of post-selection uncertainty. In designs where Gaussian assumptions hold, finite-sample calibration is exact; in more general settings, orthogonal and bootstrap-based versions offer asymptotically valid and typically conservative uncertainty quantification. We view this as a step toward integrating IM/PIM methodology with contemporary high-dimensional practice in a way that is both theoretically interpretable and practically implementable.

\bibliographystyle{agsm}
\bibliography{bibliography/references}

\appendix
\section{Appendix: Proofs and technical details}\label{app:appendix}

\subsection{Detailed statements for possibilistic validification}
\label{app:method-details}

This subsection provides the detailed technical statements of the lemmas and propositions that underpin the possibilistic validification framework in Section~2.3. We first recall the random-set formulation of inferential and possibilistic models in a simple parametric setting \citep[see, e.g.,][]{martinliu2013,martinliu2014}.

Let $\Theta$ denote the parameter space and let $Z$ denote the observed data.
An inferential random set is a random subset $\mathcal{S}(Z)$ of $\Theta$, defined on an auxiliary
probability space. Given $\mathcal{S}(Z)$, we define a belief function and a plausibility function
on assertions $A \subseteq \Theta$ by
\[
\Bel_Z(A) = P\{\mathcal{S}(Z) \subseteq A \mid Z\},
\qquad
\Pl_Z(A) = P\{\mathcal{S}(Z) \cap A \neq \varnothing \mid Z\},
\]
where the probability is with respect to the auxiliary randomness under the observed $Z$. Intuitively,
$\Bel_Z(A)$ measures how strongly the random set supports $A$, while $\Pl_Z(A)$ measures how
compatible $A$ is with the data and the random set.

For singleton assertions $\{\theta\}$, the plausibility function reduces to the \emph{plausibility contour}
\[
\pi_Z(\theta) = \Pl_Z(\{\theta\}) \in [0,1], \qquad \theta \in \Theta.
\]
We say that an inferential model is \emph{consonant} if the random sets $\mathcal{S}(Z)$ are nested,
in the sense that for any two realizations $S_1,S_2$ of $\mathcal{S}(Z)$ we have either $S_1 \subseteq S_2$
or $S_2 \subseteq S_1$. In this case the plausibility function has the representation
\[
\Pl_Z(A) = \sup_{\theta \in A} \pi_Z(\theta), \qquad A \subseteq \Theta,
\]
so that $\pi_Z$ plays the role of a possibility distribution and $\Pl_Z$ is a possibility measure in the
usual sense. Throughout the paper we restrict attention to such consonant possibilistic inferential
models (PIMs), with $\pi_Z$ as the primary object.

The following standard result establishes that the probability integral transform produces a uniform distribution.

\begin{lemma}[Probability integral transform]
\label{lem:PIT}
Suppose that, for each $\theta \in \Theta$, the pivot $T(Z,\theta)$ has a continuous cumulative
distribution function $F_\theta$ under $P_\theta$. Then
\[
U_\theta(Z) := F_\theta\{T(Z,\theta)\} \sim \Unif(0,1)
\quad \text{under $P_\theta$}.
\]
\end{lemma}

The next result shows how an arbitrary contour can be turned into a consonant inferential model.

\begin{lemma}[Contours and consonant IMs]
\label{lem:contour-consonant}
Let $\pi_Z : \Theta \to [0,1]$ be a measurable function and let $U \sim \Unif(0,1)$ be independent of~$Z$.
Define the random set
\[
\mathcal{S}(Z,U) := \{\theta \in \Theta : \pi_Z(\theta) \ge U\}.
\]
Then:
\begin{enumerate}
\item For fixed $Z$, the sets $\mathcal{S}(Z,u)$ are nested in $u$, so the corresponding inferential
model is consonant.
\item For any measurable $A \subseteq \Theta$,
\(
\Pl_Z(A) = P\{\mathcal{S}(Z,U) \cap A \neq \varnothing \mid Z\}
 = \sup_{\theta \in A} \pi_Z(\theta).
\)
\end{enumerate}
\end{lemma}

We now present the main results on validification of plausibility contours constructed from pivots.

\begin{proposition}[Exact pivot validification]
\label{prop:pivot-validity}
Suppose that, for each $\theta \in \Theta$, the pivot $T(Z,\theta)$ has a continuous cumulative
distribution function $F_\theta$ under $P_\theta$. Define
\begin{equation}
\label{eq:plaus-contour-app}
U_\theta(Z) = F_\theta\{T(Z,\theta)\}, \qquad
\pi_Z(\theta) = 1 - \bigl| 2 U_\theta(Z) - 1 \bigr|.
\end{equation}

Then:
\begin{enumerate}
\item The function $\pi_Z$ is a plausibility contour and, by \Cref{lem:contour-consonant}, the random sets
$\mathcal{S}(Z,U) = \{\theta : \pi_Z(\theta) \ge U\}$ yield a consonant IM.
\item For any $\theta_0 \in \Theta$,
\(
\pi_Z(\theta_0) \sim \Unif(0,1)
\)
under $P_{\theta_0}$. In particular, the strong-validity property holds
with equality.
\end{enumerate}
\end{proposition}

\noindent
The proof is a short application of the probability integral
transform: under $P_{\theta_0}$, $U_{\theta_0}(Z)$ is $\Unif(0,1)$, and the transformation
$u \mapsto 1 - |2u-1|$ preserves the uniform distribution.

For later use with bootstrap pivots we also state an approximate version.

\begin{proposition}[Approximate pivot validification]
\label{prop:approx-pivot}
For each sample size $n$ and parameter $\theta\in\Theta$, let $T_n(Z,\theta)$ be a scalar statistic with
true cumulative distribution function $F_{n,\theta}$ under $P_\theta$. Let $\widehat F_{n,\theta}$ be
a (possibly data-dependent) reference distribution function and define
\[
U_{n,\theta}(Z) = \widehat F_{n,\theta}\{T_n(Z,\theta)\},
\qquad
\pi_{n,Z}(\theta) = 1 - \bigl| 2 U_{n,\theta}(Z) - 1 \bigr|.
\]
Suppose that the approximation $\widehat F_{n,\theta}$ is uniformly consistent in Kolmogorov
distance, in the sense that
\[
\sup_{\theta_0 \in \Theta} \Delta_n(\theta_0)
:= \sup_{\theta_0 \in \Theta} \sup_{t \in \mathbb{R}}
  \bigl| F_{n,\theta_0}(t) - \widehat F_{n,\theta_0}(t) \bigr|
\to 0
\quad \text{as $n \to \infty$}.
\]
Then, for any $u \in [0,1]$,
\[
\sup_{\theta_0 \in \Theta} P_{\theta_0}\bigl\{ \pi_{n,Z}(\theta_0) \le u \bigr\}
\;\le\;
u + 2 \sup_{\theta_0 \in \Theta} \Delta_n(\theta_0),
\]
so that
\[
\limsup_{n\to\infty}
\sup_{\theta_0 \in \Theta}
P_{\theta_0}\bigl\{ \pi_{n,Z}(\theta_0) \le u \bigr\}
\;\le\; u, \qquad u \in [0,1].
\]
\end{proposition}

\noindent
The proof controls the deviation of
$U_{n,\theta_0}(Z)$ from $\Unif(0,1)$ in terms of the Kolmogorov distance between $F_{n,\theta_0}$
and $\widehat F_{n,\theta_0}$ and then propagates this bound through the transformation
$u \mapsto 1 - |2u-1|$.

\subsection{PIM validification and finite-sample results}
\label{app:validification}

\begin{proof}[Proof of Lemma~\ref{lem:PIT}]
Fix $\theta \in \Theta$ and write $T = T(Z,\theta)$ and $F = F_\theta$ for brevity. By assumption,
$F$ is a continuous cumulative distribution function of $T$ under $P_\theta$. For any $u \in [0,1]$,
let $F^{-1}(u)$ denote the (left-continuous) inverse of $F$. Continuity of $F$ implies that
$F\{F^{-1}(u)\} = u$. Then
\[
P_\theta\bigl\{ U_\theta(Z) \le u \bigr\}
= P_\theta\bigl\{ F(T) \le u \bigr\}
= P_\theta\bigl\{ T \le F^{-1}(u) \bigr\}
= F\{F^{-1}(u)\}
= u.
\]
Hence $U_\theta(Z)$ is $\Unif(0,1)$ under $P_\theta$, as claimed.
\end{proof}

\begin{proof}[Proof of Lemma~\ref{lem:contour-consonant}]
Fix a realization $Z$ and consider the family of sets
\[
\mathcal{S}(Z,u) = \{\theta \in \Theta : \pi_Z(\theta) \ge u\}, \qquad u \in [0,1].
\]

For part~(1), let $0 \le u_1 \le u_2 \le 1$. Then
\[
\mathcal{S}(Z,u_2)
= \{\theta : \pi_Z(\theta) \ge u_2\}
\subseteq \{\theta : \pi_Z(\theta) \ge u_1\}
= \mathcal{S}(Z,u_1),
\]
so the sets $\mathcal{S}(Z,u)$ are nested in $u$ and the resulting inferential model is consonant.

For part~(2), fix $Z$ and a measurable $A \subseteq \Theta$. Let
\[
m_A(Z) := \sup_{\theta \in A} \pi_Z(\theta) \in [0,1].
\]
We first show that, for fixed $Z$,
\[
\{\mathcal{S}(Z,U) \cap A \neq \varnothing\}
= \{ U \le m_A(Z)\}.
\]
Indeed, the event $\{\mathcal{S}(Z,U) \cap A \neq \varnothing\}$ means that there exists
$\theta \in A$ such that $\pi_Z(\theta) \ge U$, which implies $U \le m_A(Z)$. Conversely, if
$U \le m_A(Z)$ then by definition of the supremum there exists a sequence
$\{\theta_k\}_{k\ge 1} \subseteq A$ with $\pi_Z(\theta_k) \uparrow m_A(Z)$. For $k$ large enough we
have $\pi_Z(\theta_k) > U$, so $\theta_k \in \mathcal{S}(Z,U) \cap A$ and the intersection is
non-empty. This proves the equivalence of events.

Since $U \sim \Unif(0,1)$ and is independent of $Z$, conditioning on $Z$ yields
\[
\Pl_Z(A)
= P\{\mathcal{S}(Z,U) \cap A \neq \varnothing \mid Z\}
= P\{ U \le m_A(Z) \mid Z\}
= m_A(Z)
= \sup_{\theta \in A} \pi_Z(\theta),
\]
as desired.
\end{proof}

\begin{proof}[Proof of Proposition~\ref{prop:pivot-validity}]
By Lemma~\ref{lem:contour-consonant}, for any measurable contour $\pi_Z$ the random set
$\mathcal{S}(Z,U) = \{\theta : \pi_Z(\theta) \ge U\}$ with $U \sim \Unif(0,1)$ independent of $Z$
defines a consonant IM whose plausibility function satisfies
\[
\Pl_Z(A) = \sup_{\theta \in A} \pi_Z(\theta), \qquad A \subseteq \Theta.
\]
This proves part~(1) once we fix the specific contour
\[
\pi_Z(\theta) = 1 - \bigl| 2 U_\theta(Z) - 1 \bigr|, \qquad
U_\theta(Z) = F_\theta\{T(Z,\theta)\}.
\]

For part~(2), fix $\theta_0 \in \Theta$ and write $U = U_{\theta_0}(Z)$ and $\pi = \pi_Z(\theta_0)$
for brevity. By Lemma~\ref{lem:PIT}, $U \sim \Unif(0,1)$ under $P_{\theta_0}$. Define
\[
V = \bigl| 2U - 1 \bigr| \in [0,1].
\]
For any $v \in [0,1]$,
\begin{align*}
P_{\theta_0}(V \le v)
&= P_{\theta_0}\bigl( |2U-1| \le v \bigr) \\
&= P_{\theta_0}\bigl( -v \le 2U-1 \le v \bigr) \\
&= P_{\theta_0}\bigl( (1-v)/2 \le U \le (1+v)/2 \bigr) \\
&= \frac{1+v}{2} - \frac{1-v}{2}
= v,
\end{align*}
so $V \sim \Unif(0,1)$. Since $\pi = 1 - V$, another simple change of variables shows that
$\pi \sim \Unif(0,1)$ as well: for any $u \in [0,1]$,
\[
P_{\theta_0}(\pi \le u)
= P_{\theta_0}(1 - V \le u)
= P_{\theta_0}(V \ge 1-u)
= 1 - P_{\theta_0}(V < 1-u)
= 1 - (1-u)
= u.
\]
Therefore, $\pi_Z(\theta_0)$ is $\Unif(0,1)$ under $P_{\theta_0}$ and the strong-validity property
\eqref{eq:strong-validity} holds with equality.
\end{proof}

\begin{proof}[Proof of Proposition~\ref{prop:approx-pivot}]
Fix $n$ and $\theta_0$, and write $T_n = T_n(Z,\theta_0)$, $F = F_{n,\theta_0}$ and
$\widehat F = \widehat F_{n,\theta_0}$ for brevity. Let
\[
\Delta_n(\theta_0)
= \sup_{t \in \mathbb{R}} \bigl| F(t) - \widehat F(t) \bigr|.
\]

First consider the pseudo-$p$-value
\[
U_n = \widehat F(T_n).
\]
Let $U_n^0 = F(T_n)$. By Lemma~\ref{lem:PIT} applied to $T_n$ and $F$, $U_n^0 \sim \Unif(0,1)$
under $P_{\theta_0}$. For any $u \in [0,1]$ we have the inclusions
\[
\{ U_n \le u \}
= \{ \widehat F(T_n) \le u \}
\subseteq \{ F(T_n) \le u + \Delta_n(\theta_0) \}
= \{ U_n^0 \le u + \Delta_n(\theta_0) \},
\]
and
\[
\{ U_n^0 \le u - \Delta_n(\theta_0) \}
\subseteq \{ \widehat F(T_n) \le u \}
= \{ U_n \le u \},
\]
because $|\widehat F(t) - F(t)| \le \Delta_n(\theta_0)$ for all $t$. Hence
\[
P_{\theta_0}\bigl\{ U_n^0 \le u - \Delta_n(\theta_0) \bigr\}
\;\le\;
P_{\theta_0}\bigl\{ U_n \le u \bigr\}
\;\le\;
P_{\theta_0}\bigl\{ U_n^0 \le u + \Delta_n(\theta_0) \bigr\}.
\]
Since $U_n^0 \sim \Unif(0,1)$, this implies
\[
(u - \Delta_n(\theta_0))_+
\;\le\;
P_{\theta_0}\bigl\{ U_n \le u \bigr\}
\;\le\;
\min\{u + \Delta_n(\theta_0), 1\},
\]
and therefore
\[
\sup_{u \in [0,1]}
\bigl| P_{\theta_0}\{U_n \le u\} - u \bigr|
\;\le\;
\Delta_n(\theta_0).
\tag{$\ast$}
\]

Next define
\[
V_n = \bigl| 2 U_n - 1 \bigr| \in [0,1],
\qquad
\pi_{n,Z}(\theta_0) = 1 - V_n.
\]
For any $v \in [0,1]$,
\[
\{V_n \le v\}
= \bigl\{ (1-v)/2 \le U_n \le (1+v)/2 \bigr\},
\]
so
\begin{align*}
P_{\theta_0}(V_n \le v)
&= P_{\theta_0}\bigl\{ U_n \le (1+v)/2 \bigr\}
  - P_{\theta_0}\bigl\{ U_n < (1-v)/2 \bigr\} \\
&\le \Bigl( \frac{1+v}{2} + \Delta_n(\theta_0) \Bigr)
    - \Bigl( \frac{1-v}{2} - \Delta_n(\theta_0) \Bigr)
 = v + 2\Delta_n(\theta_0),
\end{align*}
where we used $(\ast)$ at the two endpoints. Similarly,
\begin{align*}
P_{\theta_0}(V_n \le v)
&\ge \Bigl( \frac{1+v}{2} - \Delta_n(\theta_0) \Bigr)
    - \Bigl( \frac{1-v}{2} + \Delta_n(\theta_0) \Bigr)
 = v - 2\Delta_n(\theta_0).
\end{align*}
Hence
\[
\sup_{v \in [0,1]}
\bigl| P_{\theta_0}(V_n \le v) - v \bigr|
\;\le\;
2\Delta_n(\theta_0).
\]

Finally, for $u \in [0,1]$,
\[
P_{\theta_0}\bigl\{ \pi_{n,Z}(\theta_0) \le u \bigr\}
= P_{\theta_0}\bigl\{ 1 - V_n \le u \bigr\}
= P_{\theta_0}\bigl\{ V_n \ge 1-u \bigr\}
= 1 - P_{\theta_0}\bigl\{ V_n < 1-u \bigr\}.
\]
Since the distribution function of $V_n$ is within $2\Delta_n(\theta_0)$ of the uniform cdf on
$[0,1]$, it follows that
\[
P_{\theta_0}\bigl\{ \pi_{n,Z}(\theta_0) \le u \bigr\}
\;\le\;
u + 2\Delta_n(\theta_0), \qquad u \in [0,1].
\]
Taking the supremum over $\theta_0\in\Theta$ yields
\[
\sup_{\theta_0\in\Theta}
P_{\theta_0}\bigl\{ \pi_{n,Z}(\theta_0) \le u \bigr\}
\;\le\;
u + 2 \sup_{\theta_0\in\Theta} \Delta_n(\theta_0),
\]
and the limit
\[
\limsup_{n\to\infty}\,
\sup_{\theta_0\in\Theta} P_{\theta_0}\{\pi_{n,Z}(\theta_0) \le u\}\le u
\]
follows from the uniform assumption $\sup_{\theta_0\in\Theta}\Delta_n(\theta_0) \to 0$.
\end{proof}

\subsection{Illustrative example: single-mean $t$-interval}
\label{app:singlemean}

For completeness we record the single-mean example used to illustrate the generic PIM validification recipe. Consider i.i.d.\ observations $Y_1,\ldots,Y_n \sim N(\theta,\sigma^2)$ with unknown $\sigma^2$. Let $\bar Y$ and $\hat\sigma$ denote the sample mean and standard deviation, and define the usual $t$-statistic
\[
T(Z,\theta) = \frac{\sqrt{n}(\bar Y - \theta)}{\hat\sigma}, \qquad Z = (Y_1,\ldots,Y_n).
\]
Under $\theta$, $T(Z,\theta)$ has a $t_{n-1}$ distribution with distribution function $F_{n-1}$, so Proposition~\ref{prop:pivot-validity} applies. The contour becomes
\[
\pi_Z(\theta) = 1 - \bigl| 2 F_{n-1}\{ T(Z,\theta) \} - 1 \bigr|,
\]
which is maximized at $\theta = \bar Y$ and decreases symmetrically as $|\theta - \bar Y|$ increases. The $(1-\alpha)$ upper-level set $\{\theta : \pi_Z(\theta) \ge \alpha\}$ is equivalent to $\{|T(Z,\theta)| \le t_{1-\alpha/2,n-1}\}$ and therefore coincides with the classical Student $t$ interval
\[
\bar Y \pm t_{1-\alpha/2,n-1} \frac{\hat\sigma}{\sqrt{n}}.
\]
Thus the PIM construction does not alter the numerical form of the interval; it organizes the classical $t$-interval as the $(1-\alpha)$ upper-level set of a strongly valid plausibility contour.

\subsection{Orthogonalized extension: partially linear model}
\label{app:pl-orth}

In this subsection we verify Lemma~\ref{lem:pl-conditions} by checking Assumptions~\textup{(O1)}–\textup{(O4)}
under Conditions~\textup{(PL1)}–\textup{(PL3)} for the partially linear model.

Recall the partially linear model
\[
  Y = D\theta_0 + g_0(X) + \varepsilon,\qquad
  D = m_0(X) + v,
\]
with i.i.d.\ observations $W=(Y,D,X)$, where $\varepsilon$ and $v$ are mean-zero errors,
independent of $X$ and with finite moments as specified in~\textup{(PL1)}. The orthogonal score
for $\theta$ is
\[
  \psi(W;\theta,\eta) =
  \bigl\{D - m(X)\bigr\}\bigl\{Y - g(X) - \theta\{D - m(X)\}\bigr\},
  \qquad \eta=(g,m),
\]
and the target value $\theta_0$ satisfies
$\mathbb{E}\{\psi(W;\theta_0,\eta_0)\mid\mathcal{G}\}=0$ with $\eta_0=(g_0,m_0)$.

We first record the orthogonality property.

\begin{lemma}[Neyman orthogonality of the partially linear score]
\label{lem:pl-orthogonality}
Under the partially linear model and~\textup{(PL1)}, the score $\psi(W;\theta,\eta)$ is
Gateaux differentiable in $(\theta,\eta)$ and satisfies Assumption~\textup{(O1)} with
$\eta_0=(g_0,m_0)$.
\end{lemma}

\begin{proof}
By construction,
\[
  \mathbb{E}\{\psi(W;\theta_0,\eta_0)\mid\mathcal{G}\}
  = \mathbb{E}\bigl[(D-m_0(X))\{\varepsilon + g_0(X) - g_0(X)\}
    - (D-m_0(X))^2(\theta_0-\theta_0)\mid\mathcal{G}\bigr]=0.
\]
For Gateaux differentiability, consider perturbations
$\eta_t = (g_0 + t h_g, m_0 + t h_m)$ with bounded measurable directions
$h_g,h_m$. A direct differentiation under the expectation sign yields
\[
  \partial_\eta \mathbb{E}\{\psi(W;\theta_0,\eta)\mid\mathcal{G}\}\big|_{\eta=\eta_0}
  = \mathbb{E}\bigl[(D-m_0(X))\{-h_g(X) + \theta_0 h_m(X)\}\mid\mathcal{G}\bigr].
\]
Using $D = m_0(X) + v$ and $\mathbb{E}(v\mid X)=0$ gives
\[
  \mathbb{E}\bigl[(D-m_0(X))h_g(X)\mid\mathcal{G}\bigr]
  = \mathbb{E}\bigl[v h_g(X)\mid\mathcal{G}\bigr] = 0,
\]
and similarly
$\mathbb{E}\bigl[(D-m_0(X))h_m(X)\mid\mathcal{G}\bigr]=0$, so the derivative vanishes.
Local smoothness in a neighborhood of $(\theta_0,\eta_0)$ follows from the moment bounds
in~\textup{(PL1)} and standard dominated-convergence arguments. Hence Assumption~\textup{(O1)}
holds.
\end{proof}

We next derive $L_2$-rates for the cross-fitted nuisance estimators.

\begin{lemma}[Lasso prediction rates for $g_0$ and $m_0$]
\label{lem:pl-lasso}
Suppose Conditions~\textup{(PL1)}–\textup{(PL3)} hold. On each inference fold
$b = 1,\ldots,B$, let $\widehat g^{(-b)}$ and $\widehat m^{(-b)}$ be $\ell_1$-penalized estimators
of $g_0$ and $m_0$ trained on the complementary folds with penalty levels
$\lambda_{g,n},\lambda_{m,n}\asymp \sqrt{(\log p_n)/n_{\mathrm{inf}}}$. Then there exists a constant
$C<\infty$ such that, uniformly over folds $b$ and over $\beta_0\in\mathcal{B}_{s_n}$,
\[
  \bigl\|\widehat g^{(-b)} - g_0\bigr\|_{L_2(P_X)}
  \le C\Bigl\{\frac{s_n \log p_n}{n_{\mathrm{inf}}}\Bigr\}^{1/2},
  \qquad
  \bigl\|\widehat m^{(-b)} - m_0\bigr\|_{L_2(P_X)}
  \le C\Bigl\{\frac{s_n \log p_n}{n_{\mathrm{inf}}}\Bigr\}^{1/2},
\]
and, in particular,
\[
  \bigl\|\widehat g^{(-b)} - g_0\bigr\|_{L_2(P_X)} = o_p(n_{\mathrm{inf}}^{-1/4}),
  \qquad
  \bigl\|\widehat m^{(-b)} - m_0\bigr\|_{L_2(P_X)} = o_p(n_{\mathrm{inf}}^{-1/4}).
\]
\end{lemma}

\begin{proof}
By~\textup{(PL2)} the functions $g_0$ and $m_0$ admit sparse linear approximations in the
dictionary of regressors with sparsity $s_n$, while~\textup{(PL3)} provides a compatibility or
restricted-eigenvalue condition for the design. Standard oracle-inequality arguments for the
lasso (see, e.g., \citealp{buhlmann2011statistics,zhang2008sparsity}) then yield, for suitable
penalties $\lambda_{g,n},\lambda_{m,n}$,
\[
  \bigl\|\widehat g^{(-b)} - g_0\bigr\|_{L_2(P_X)}
  \lesssim \Bigl\{\frac{s_n \log p_n}{n_{\mathrm{inf}}}\Bigr\}^{1/2},
  \qquad
  \bigl\|\widehat m^{(-b)} - m_0\bigr\|_{L_2(P_X)}
  \lesssim \Bigl\{\frac{s_n \log p_n}{n_{\mathrm{inf}}}\Bigr\}^{1/2},
\]
uniformly over $\beta_0\in\mathcal{B}_{s_n}$ and folds $b$. Assumption~\textup{(PL3)} ensures
$s_n^2 (\log p_n)^2 / n_{\mathrm{inf}} \to 0$, so
\[
  \Bigl\{\frac{s_n \log p_n}{n_{\mathrm{inf}}}\Bigr\}^{1/2}
  = n_{\mathrm{inf}}^{-1/4}
    \Bigl\{\frac{s_n \log p_n}{n_{\mathrm{inf}}^{1/2}}\Bigr\}^{1/2}
  = o(n_{\mathrm{inf}}^{-1/4}),
\]
which proves the stated $o_p(n_{\mathrm{inf}}^{-1/4})$ rates.
\end{proof}

We now establish the asymptotic linear representation and central limit theorem.

\begin{lemma}[Asymptotic linearity and CLT for $T_n(\theta_0)$]
\label{lem:pl-clt}
Under Conditions~\textup{(PL1)}–\textup{(PL3)}, the cross-fitted score satisfies
\[
  \sqrt{n_{\mathrm{inf}}}\,\widehat U_n(\theta_0)
  = \frac{1}{\sqrt{n_{\mathrm{inf}}}}
    \sum_{i\in I_{\mathrm{inf}}} \psi(W_i;\theta_0,\eta_0)
    + o_p(1),
\]
and, conditional on $\mathcal{G}$,
\[
  \frac{1}{\sqrt{n_{\mathrm{inf}}}}
    \sum_{i\in I_{\mathrm{inf}}} \psi(W_i;\theta_0,\eta_0)
  \;\xRightarrow{d}\; N(0,V_\psi),
\]
where $V_\psi = \Var\{\psi(W;\theta_0,\eta_0)\mid\mathcal{G}\}$ is positive and finite. If
$\widehat\sigma^2(\theta)$ consistently estimates $V_\psi$ at $\theta_0$, then
$T_n(\theta_0)$ defined in~\eqref{eq:orth-score} converges in law to $N(0,1)$ conditional on
$\mathcal{G}$.
\end{lemma}

\begin{proof}
By Lemma~\ref{lem:pl-lasso} and the orthogonality in Lemma~\ref{lem:pl-orthogonality}, a
first-order Taylor expansion of
$\psi(W;\theta_0,\widehat\eta^{(-b)})$ around $\eta_0$ yields
\[
  \psi(W;\theta_0,\widehat\eta^{(-b)})
  = \psi(W;\theta_0,\eta_0) + R_n(W),
\]
where the remainder $R_n(W)$ satisfies
\[
  \mathbb{E}\bigl[|R_n(W)| \mid \mathcal{G}\bigr]
  \lesssim \bigl\|\widehat\eta^{(-b)} - \eta_0\bigr\|_{L_2(P_X)}^2
  = o(n_{\mathrm{inf}}^{-1/2}).
\]
Averaging over $i\in I_{\mathrm{inf}}$ and folds $b$ shows that the contribution of the
remainder to $\sqrt{n_{\mathrm{inf}}}\,\widehat U_n(\theta_0)$ is $o_p(1)$, establishing the
asymptotic linear representation. The central limit theorem for
$n_{\mathrm{inf}}^{-1/2}\sum_{i\in I_{\mathrm{inf}}} \psi(W_i;\theta_0,\eta_0)$ follows from the
i.i.d.\ structure and the finite $(2+\delta)$th moment in~\textup{(PL1)} via the Lindeberg–Feller
or Lyapunov criterion. Consistency of $\widehat\sigma^2(\theta_0)$ and Slutsky's theorem then
yield the $N(0,1)$ limit for $T_n(\theta_0)$ conditional on $\mathcal{G}$.
\end{proof}

Finally we state the bootstrap approximation.

\begin{lemma}[Wild bootstrap for the orthogonal score]
\label{lem:pl-bootstrap}
Let $T_n^*(\theta_0)$ be the multiplier wild-bootstrap version of $T_n(\theta_0)$ constructed
from independent mean-zero, unit-variance multipliers with finite $(2+\delta)$th moments as in
Section~\ref{sec:computation}. Under Conditions~\textup{(PL1)}–\textup{(PL3)}, we have
\[
  \sup_{t\in\mathbb{R}}
  \bigl|
    P\{T_n(\theta_0)\le t \mid \mathcal{G}\}
    - P^*\{T_n^*(\theta_0)\le t \mid \mathcal{G}\}
  \bigr|
  \;\to\; 0
\]
in probability as $n\to\infty$.
\end{lemma}

\begin{proof}
Write $\varphi_i = \psi(W_i;\theta_0,\eta_0)$ and let
$V_\psi = \Var\{\varphi_1 \mid \mathcal{G}\} \in (0,\infty)$.
By Lemma~\ref{lem:pl-clt} there exists a remainder term $r_n=o_P(1)$ such that
\[
  \sqrt{n_{\mathrm{inf}}}\,\widehat U_n(\theta_0)
  = \frac{1}{\sqrt{n_{\mathrm{inf}}}}
    \sum_{i\in I_{\mathrm{inf}}} \varphi_i + r_n,
\]
and, conditional on $\mathcal{G}$,
\[
  \frac{1}{\sqrt{n_{\mathrm{inf}}}}
    \sum_{i\in I_{\mathrm{inf}}} \varphi_i
  \;\xRightarrow{d}\; N(0,V_\psi).
\]
Moreover, $\widehat\sigma^2(\theta_0)$ consistently estimates $V_\psi$; hence, by Slutsky’s
theorem,
\[
  T_n(\theta_0)
  = \frac{\sqrt{n_{\mathrm{inf}}}\,\widehat U_n(\theta_0)}{\widehat\sigma(\theta_0)}
  \;\xRightarrow{d}\; N(0,1)
  \qquad\text{conditionally on $\mathcal{G}$}.
\]

For the bootstrap statistic, define the centered influence values
\[
  \varphi_i^\circ = \varphi_i - \bar\varphi_n,
  \qquad
  \bar\varphi_n = \frac{1}{n_{\mathrm{inf}}}
  \sum_{i\in I_{\mathrm{inf}}} \varphi_i,
\]
and consider
\[
  T_n^*(\theta_0)
  = \frac{1}{\sqrt{n_{\mathrm{inf}}}\,\widehat\sigma^*(\theta_0)}
    \sum_{i\in I_{\mathrm{inf}}} \xi_i \varphi_i^\circ,
\]
where $\{\xi_i\}_{i\in I_{\mathrm{inf}}}$ are i.i.d.\ multipliers with
$\mathbb{E}(\xi_i)=0$, $\Var(\xi_i)=1$ and $\mathbb{E}|\xi_i|^{2+\delta}<\infty$,
independent of the data, and $\widehat\sigma^{*2}(\theta_0)$ is the bootstrap variance
estimate.

Conditional on the data (and hence on $\mathcal{G}$), the random variables
$\{\xi_i \varphi_i^\circ : i\in I_{\mathrm{inf}}\}$ are independent with mean zero and
\[
  \Var\Bigl(
    \frac{1}{\sqrt{n_{\mathrm{inf}}}}
    \sum_{i\in I_{\mathrm{inf}}} \xi_i \varphi_i^\circ
    \,\Big|\, \mathcal{G}
  \Bigr)
  = \frac{1}{n_{\mathrm{inf}}}
    \sum_{i\in I_{\mathrm{inf}}} \varphi_i^{\circ 2}
  \;\xrightarrow{P}\; V_\psi,
\]
by the law of large numbers and the finite second moment implied by~\textup{(PL1)}.

To verify a Lyapunov-type condition, note that, conditional on $\mathcal{G}$,
\[
  \mathbb{E}\bigl(|\xi_i \varphi_i^\circ|^{2+\delta} \mid \mathcal{G}\bigr)
  = \mathbb{E}|\xi_i|^{2+\delta}\,|\varphi_i^\circ|^{2+\delta}
  \lesssim |\varphi_i^\circ|^{2+\delta}.
\]
Hence
\[
  \frac{1}{n_{\mathrm{inf}}^{1+\delta/2}}
  \sum_{i\in I_{\mathrm{inf}}}
  \mathbb{E}\bigl(|\xi_i \varphi_i^\circ|^{2+\delta} \mid \mathcal{G}\bigr)
  \lesssim
  \frac{1}{n_{\mathrm{inf}}^{1+\delta/2}}
  \sum_{i\in I_{\mathrm{inf}}} |\varphi_i^\circ|^{2+\delta}.
\]
Assumption~\textup{(PL1)} guarantees $\mathbb{E}|\varphi_i|^{2+\delta}<\infty$, so by the
law of large numbers,
\[
  \frac{1}{n_{\mathrm{inf}}}
  \sum_{i\in I_{\mathrm{inf}}} |\varphi_i^\circ|^{2+\delta}
  = O_P(1),
\]
and therefore
\[
  \frac{1}{n_{\mathrm{inf}}^{1+\delta/2}}
  \sum_{i\in I_{\mathrm{inf}}} |\varphi_i^\circ|^{2+\delta}
  = n_{\mathrm{inf}}^{-\delta/2}
    \Bigl\{\frac{1}{n_{\mathrm{inf}}}
      \sum_{i\in I_{\mathrm{inf}}} |\varphi_i^\circ|^{2+\delta}\Bigr\}
  \;\xrightarrow{P}\; 0.
\]
This establishes the Lyapunov condition (conditionally on $\mathcal{G}$) and yields
\[
  \frac{1}{\sqrt{n_{\mathrm{inf}}}}
    \sum_{i\in I_{\mathrm{inf}}} \xi_i \varphi_i^\circ
  \;\xRightarrow{d}\; N(0,V_\psi)
  \qquad\text{conditionally on $\mathcal{G}$}.
\]
Furthermore, $\widehat\sigma^{*2}(\theta_0)$ consistently estimates $V_\psi$, again by the
law of large numbers and the moment bound in~\textup{(PL1)}. Slutsky’s theorem then implies
\[
  T_n^*(\theta_0)
  \;\xRightarrow{d}\; N(0,1)
  \qquad\text{conditionally on $\mathcal{G}$}.
\]

Let $\Phi$ denote the standard normal distribution function. The two conditional central limit
theorems above imply
\[
  \sup_{t\in\mathbb{R}}
  \bigl| P\{T_n(\theta_0)\le t \mid \mathcal{G}\} - \Phi(t) \bigr|
  \;\to\; 0,
  \qquad
  \sup_{t\in\mathbb{R}}
  \bigl| P^*\{T_n^*(\theta_0)\le t \mid \mathcal{G}\} - \Phi(t) \bigr|
  \;\to\; 0
\]
in probability. By the triangle inequality,
\[
  \sup_{t\in\mathbb{R}}
  \bigl|
    P\{T_n(\theta_0)\le t \mid \mathcal{G}\}
    - P^*\{T_n^*(\theta_0)\le t \mid \mathcal{G}\}
  \bigr|
\]
\[
  \;\le\;
\]
\[
  \sup_t |P\{T_n(\theta_0)\le t \mid \mathcal{G}\}-\Phi(t)|
  + \sup_t |\Phi(t)-P^*\{T_n^*(\theta_0)\le t \mid \mathcal{G}\}|
  \;\to\; 0
\]
in probability. This is the desired Kolmogorov-distance convergence and verifies
Assumption~\textup{(O4)} for the partially linear model.
\end{proof}

\begin{proof}[Proof of Lemma~\ref{lem:pl-conditions}]
Assumption~\textup{(O1)} is exactly Lemma~\ref{lem:pl-orthogonality}. Lemma~\ref{lem:pl-lasso}
implies the $L_2$-rate condition and stochastic equicontinuity in~\textup{(O2)}. The asymptotic
linearity and central limit theorem in~\textup{(O3)} follow from Lemma~\ref{lem:pl-clt}. Finally,
Lemma~\ref{lem:pl-bootstrap} yields the bootstrap approximation in~\textup{(O4)}. This proves
Lemma~\ref{lem:pl-conditions}.
\end{proof}

\subsection{High-dimensional conditions for RSPIM}\label{app:HD}

\medskip
\noindent\textbf{Assumptions (HD).}
\begin{itemize}
\item[(HD1)] \emph{Sparse linear model.} The full-data model is $Y = X\beta_0+\varepsilon$ with
$\varepsilon\sim N(0,\sigma^2 I_n)$, $\beta_0\in\mathbb{R}^{p_n}$, and $s_n = |S_0|$ satisfying
$s_n\log p_n/n\to 0$ as $n\to\infty$.

\item[(HD2)] \emph{Design regularity.} The rows of $X$ are i.i.d.\ sub-Gaussian with mean zero and
covariance matrix $\Sigma$. There exist constants
$0<\kappa_{\min}\le \kappa_{\max}<\infty$ and $C_0<\infty$ such that the eigenvalues of $\Sigma$
are bounded away from zero and infinity on $s_n$-sparse directions:
\[
\kappa_{\min} \le \frac{v^\top \Sigma v}{\|v\|_2^2} \le \kappa_{\max}
\qquad\text{for all }v\in\mathbb{R}^{p_n}\text{ with }\|v\|_0 \le C_0 s_n.
\]
Moreover, with probability tending to one,
\[
\kappa_{\min} \le \lambda_{\min}\!\Bigl(\frac{1}{n}X_T^\top X_T\Bigr)
\le \lambda_{\max}\!\Bigl(\frac{1}{n}X_T^\top X_T\Bigr)\le \kappa_{\max}
\]
simultaneously for all index sets $T\subset\{1,\dots,p_n\}$ with $|T|\le C_0 s_n$.
In particular, OLS refits on subsets of size $O(s_n)$ are well defined and the corresponding
Gram matrices are uniformly well-conditioned with high probability.

\item[(HD3)] \emph{Split proportions.} The selection and inference subsamples satisfy
$|I_{\mathrm{sel}}|\asymp n$ and $|I_{\mathrm{inf}}|\asymp n$ as $n\to\infty$.

\item[(HD4)] \emph{High-dimensional screening selector.} For each sample size $n$, let
$\mathcal{S}_n$ denote the split-based selector applied to $I_{\mathrm{sel}}$ and write
$\widehat S_n$ for its selected set. There exists a constant $C<\infty$ such that
\[
\mathbb{P}_{\beta_0}\bigl(S_0\subseteq \widehat S_n,\;|\widehat S_n|\le C s_n\bigr)\to 1,
\]
uniformly over $\beta_0\in\mathcal{B}_{s_n}$ as $n\to\infty$.
\end{itemize}

Assumption~(HD4) is a high-level screening condition: it requires the selector to include the true
support and to return a set whose size is of the same order as $s_n$, uniformly over the sparse class
$\mathcal{B}_{s_n}$. Many concrete selectors satisfy this property under suitable design and sparsity
conditions; we verify it for lasso stability selection below. Since $s_n = o(n_{\mathrm{inf}})$ by
Assumption~\textup{(HD1)}, there exist $\kappa \in (0,1)$ and $N < \infty$ such that, on the event in
Assumption~\textup{(HD4)}, 
\[
|\widehat S_n| \le C s_n \le \kappa n_{\mathrm{inf}}
\]
for all $n \ge N$. We use this observation in the proof of Theorem~\ref{thm:hd-strong} to ensure that
refits on $I_{\mathrm{inf}}$ remain low-dimensional relative to $n_{\mathrm{inf}}$.

On the event $\{S_0 \subseteq \widehat S_n,\;|\widehat S_n|\le C s_n\}$, inference
reduces to a fixed-dimensional Gaussian linear model on $I_{\mathrm{inf}}$ to which Assumptions
\textup{(L1)}–\textup{(L2)} apply. In particular, conditional on $\mathcal{G}$ the refitted model satisfies the
finite-sample strong-validity result of Theorem~\ref{thm:finite-sample-strong} for any fixed
support contained in $\widehat S_n$. For any active coordinate $j\in S_0$ we have $j\in\widehat S_n$ on
this event, so the post-selection interval $C_{j,n}(\alpha)$ is well defined.

We collect high-level conditions for lasso stability selection on $I_{\mathrm{sel}}$.

\medskip
\noindent\textbf{Conditions (C) for lasso stability selection.}
\begin{itemize}
\item[(C1)] \emph{Sub-Gaussian design and restricted eigenvalue.} The rows of $X$ are i.i.d.\
sub-Gaussian with covariance matrix $\Sigma$ whose eigenvalues are bounded away from zero and
infinity. The design satisfies a sparse eigenvalue / restricted eigenvalue (or sparse Riesz)
condition and related regularity assumptions ensuring the usual oracle inequalities, sparsity, and
screening properties of the lasso, as in standard high-dimensional regression theory.

\item[(C2)] \emph{Beta-min condition.} There exists a constant $c_\beta>0$ such that
$\min_{j\in S_0}|\beta_{0j}|\ge c_\beta\sigma\sqrt{(\log p_n)/n}$ and $s_n\to\infty$. Combined
with the sparsity scaling $s_n\log p_n/n\to 0$ in Assumption~(HD1), this allows the nonzero
signals to shrink with $n$ while remaining above the usual detection boundary.

\item[(C3)] \emph{Regularization and stability-selection tuning.} On each subsample used for
selection we run the lasso with penalty $\lambda_n\asymp\sigma\sqrt{(2\log p_n)/n_{\mathrm{sel}}}$
and select at most $q_n$ variables, with $q_n\lesssim s_n$. Stability selection is implemented by
half-sampling $I_{\mathrm{sel}}$, recomputing the lasso on each subsample, and retaining variables
whose empirical selection probability exceeds a threshold $\pi_{\mathrm{thr}}\in(1/2,1)$. The tuning
parameters are chosen so that the per-family error rate (PFER) bound is uniformly bounded:
$\sup_n \mathrm{PFER}_n<\infty$, where $\mathrm{PFER}_n$ denotes the expected number of noise
variables selected by stability selection.
\end{itemize}

These conditions are chosen so that the base lasso enjoys an oracle inequality together with
sure-screening and sparsity properties, and stability selection controls the number of spurious
variables.

\subsection{Verification of Assumption~(HD4) for lasso stability selection}\label{app:hd4-lasso}

We now explain how Assumption~\textup{(HD4)} follows from standard results when the selector on
$I_{\mathrm{sel}}$ is given by lasso combined with stability selection, under the high-dimensional
regime (HD1)–(HD3) and the lasso design and tuning Conditions~\textup{(C1)}--\textup{(C3)} Section~\ref{sec:theory-hd}.
In particular, Assumption~\textup{(HD2)} provides a sparse-eigenvalue condition for the design that
underlies the lasso oracle inequalities used below.

Let $\widehat\beta^{\mathrm{lasso}}$ denote the lasso estimator on $I_{\mathrm{sel}}$ with penalty
$\lambda_n\asymp \sigma\sqrt{(2\log p_n)/n_{\mathrm{sel}}}$, and write
$S^{\mathrm{lasso}}=\mathrm{supp}(\widehat\beta^{\mathrm{lasso}})$ for its active set.

\begin{lemma}[Lasso screening and sparsity]
\label{lem:lasso-screening}
Suppose Assumptions (HD1)–(HD3) and Conditions (C1)–(C2) hold. Then there exists a constant
$C_1 < \infty$ such that, uniformly over $\beta_0 \in \mathcal{B}_{s_n}$,
\[
P_{\beta_0}\bigl( S_0 \subseteq S_{\mathrm{lasso}},\; |S_{\mathrm{lasso}}| \le C_1 s_n \bigr) \to 1.
\]
\end{lemma}

\begin{proof}
Under the sub-Gaussian design and restricted eigenvalue conditions in (C1), the lasso with
penalty level $\lambda_n \asymp \sigma \sqrt{(2 \log p_n)/n_{\sel}}$ enjoys the usual oracle inequalities
for prediction and estimation error\citet{buhlmann2011statistics}.
Combined with the beta-min condition (C2), these inequalities imply that, with probability
tending to one and uniformly over $\beta_0 \in \mathcal{B}_{s_n}$, all active coordinates are selected
(sure screening) and the lasso support has cardinality at most a constant multiple of $s_n$.
\end{proof}

Stability selection proceeds by drawing many subsamples of $I_{\mathrm{sel}}$, recomputing the lasso on each subsample
with the same $\lambda_n$, and retaining those variables whose empirical selection frequency exceeds a threshold
$\pi_{\mathrm{thr}}\in(1/2,1)$. Let $V_n$ denote the number of noise variables in the final stability-selected set $\widehat S$.

\begin{lemma}[PFER control for stability selection]
\label{lem:stability-PFER}
Suppose Condition (C3) holds. Let $V_n$ denote the number of noise variables in the final
stability-selected set $\hat S_n$. Then the per-family error rate of stability selection satisfies
\[
\mathrm{PFER}_n := \mathbb{E}(V_n)
\;\lesssim\;
\frac{q_n^2}{(2\pi_{\mathrm{thr}} - 1)\,(p_n - s_n)},
\]
where $q_n$ is the expected number of selected variables per subsample. In particular, choosing
$q_n \lesssim s_n$ and $\pi_{\mathrm{thr}} > 1/2$ yields $\sup_n \mathrm{PFER}_n < \infty$.
\end{lemma}

\begin{proof}
The displayed bound is the standard per-family error rate control for subsampling-based stability
selection and follows from the counting argument in \citet[Theorem~1]{MeinshausenBuhlmann2010},
which bounds the expected number of null variables whose empirical selection frequency exceeds
the threshold $\pi_{\mathrm{thr}}$ in terms of $q_n$ and the number of null coordinates.
\end{proof}

We now combine these two lemmas to verify Assumption~\textup{(HD4)} for the lasso stability-selection selector.

\begin{proof}[Proof of Proposition~\ref{prop:hd4-lasso}]
By Lemma~\ref{lem:lasso-screening}, under Assumptions~\textup{(HD1)}--\textup{(HD3)} and Conditions~\textup{(C1)}--\textup{(C2)},
the base lasso support $S^{\mathrm{lasso}}$ satisfies
\[
\mathbb{P}_{\beta_0}\bigl(S_0\subseteq S^{\mathrm{lasso}},\ |S^{\mathrm{lasso}}|\le C_1 s_n\bigr)\to 1
\]
uniformly over $\beta_0\in\mathcal{B}_{s_n}$, for some $C_1<\infty$.
In particular, each active variable $j\in S_0$ is selected by the base lasso with probability tending to one.

By Lemma~\ref{lem:stability-PFER}, the number $V_n$ of noise variables in the stability-selected set $\widehat S$
satisfies $\sup_n \mathbb{E}(V_n)\le C_{\mathrm{PFER}}<\infty$ for some finite constant $C_{\mathrm{PFER}}$.
Since $s_n\to\infty$ in the sparse high-dimensional regime and $V_n\le|\widehat S\setminus S_0|$, Markov's inequality yields
\[
\mathbb{P}\bigl(|\widehat S|\le C s_n\bigr)
\;\ge\; 1 - \mathbb{P}\bigl(V_n> C s_n\bigr)
\;\ge\; 1 - \frac{\mathbb{E}(V_n)}{C s_n}
\;\to\; 1
\]
for some $C\ge C_1$.
Moreover, if each active variable is selected by the base lasso with probability tending to one, then its empirical
selection frequency converges to one, so it exceeds any fixed $\pi_{\mathrm{thr}}<1$ with probability tending to one.
Therefore $S_0\subseteq \widehat S$ with probability tending to one.

Combining these two pieces, we obtain, uniformly over $\beta_0\in\mathcal{B}_{s_n}$,
\[
\mathbb{P}_{\beta_0}\bigl(S_0\subseteq \widehat S,\ |\widehat S|\le C s_n\bigr)\to 1,
\]
which is precisely Assumption~\textup{(HD4)}.
\end{proof}

\subsection{Proofs for Section 3 (Theory)}\label{app:theory-proofs}

This subsection collects the detailed assumptions and proofs for the theoretical results in Section~\ref{sec:theory}.

\subsubsection{Assumptions and proofs for finite-sample results}

We first state the assumptions for the finite-sample Gaussian linear model.

\medskip
\noindent\textbf{Assumptions (L).}
\begin{itemize}
\item[(L1)] \emph{Gaussian linear model on the inference sample.}
Conditional on $\mathcal{G}$, the response on $I_{\mathrm{inf}}$ satisfies
\[
Y_{I_{\mathrm{inf}}} = X_{I_{\mathrm{inf}}}\beta_0 + \varepsilon,\qquad
\varepsilon \sim N(0,\sigma^2 I_{n_{\mathrm{inf}}}),
\]
where the dimension $d := |\widehat S|$ satisfies $d < n_{\mathrm{inf}}$ (possibly depending on $n$).

\item[(L2)] \emph{Sample splitting and selection.}
The split $(I_{\mathrm{sel}}, I_{\mathrm{inf}})$ is independent of the errors $\varepsilon$, and the selector
uses only the data on $I_{\mathrm{sel}}$, so that conditioning on $\mathcal{G}$ fixes both the design and the
support used for inference.
\end{itemize}

\begin{proof}[Proof of Lemma~\ref{lem:selector-pivot}]
Conditional on $\mathcal{G}$, the refitted model on $I_{\mathrm{inf}}$ is a fixed-design Gaussian linear
model with support $S$ and dimension $|S|<n_{\mathrm{inf}}$. The usual OLS identities imply that
the studentized coefficients and likelihood-ratio statistics follow $t_\nu$ and $F_{d,\nu}$ distributions,
respectively. These distributions depend only on $(X_{I_{\mathrm{inf}}},S)$, which are fixed under
$\mathcal{G}$, and therefore are unaffected by the choice of selector on $I_{\mathrm{sel}}$.
\end{proof}

\begin{proof}[Proof of Theorem~\ref{thm:finite-sample-strong}]
Fix $n$ and condition on $\mathcal{G}$. Under Assumptions~\textup{(L1)}–\textup{(L2)}, 
Lemma~\ref{lem:selector-pivot} implies that the refitted model on $I_{\mathrm{inf}}$ is a fixed-design 
Gaussian linear model with support $S = \widehat S$ and dimension $d = |S| < n_{\mathrm{inf}}$, and 
the $t/F$ statistics constructed in Section~\ref{sec:method} are exact pivots conditional on 
$\mathcal{G}$. Applying Proposition~\ref{prop:pivot-validity} to these pivots, conditional on 
$\mathcal{G}$, we obtain a plausibility contour $\pi_Z(\theta_S)$ for $\theta_S = \beta_S$ such that, 
for any $\theta_{0,S}$,
\[
\pi_Z(\theta_{0,S}) \sim \Unif(0,1)
\qquad\text{conditionally on $\mathcal{G}$},
\]
and hence, for any $u \in [0,1]$,
\[
P_{\beta_0}\bigl\{ \pi_Z(\beta_{0,S}) \le u \,\big|\, \mathcal{G}\bigr\} = u.
\]

Let $C_Z(\alpha) = \{\theta_S : \pi_Z(\theta_S) > \alpha\}$ denote the $(1-\alpha)$ upper-level
set of this contour. Then for any $\beta_{0,S}$,
\[
P_{\beta_0}\bigl\{ \beta_{0,S} \notin C_Z(\alpha) \,\big|\, \mathcal{G}\bigr\}
= P_{\beta_0}\bigl\{ \pi_Z(\beta_{0,S}) \le \alpha \,\big|\, \mathcal{G}\bigr\}
\le \alpha.
\]
The coordinate-wise RSPIM interval $C_{j,n}(\alpha)$ for $\beta_{0j}$ is the projection of
$C_Z(\alpha)$ onto the $j$th coordinate:
\[
C_{j,n}(\alpha)
= \bigl\{ \theta_j : \exists\,\theta_S \text{ with } \theta_j \in \mathbb{R},\,
\theta_S \in C_Z(\alpha)\bigr\}.
\]
Hence the non-coverage event $\{\beta_{0j} \notin C_{j,n}(\alpha)\}$ is contained in
$\{\beta_{0,S} \notin C_Z(\alpha)\}$, and therefore
\[
P_{\beta_0}\bigl\{ \beta_{0j} \notin C_{j,n}(\alpha) \,\big|\, \mathcal{G}\bigr\}
\le
P_{\beta_0}\bigl\{ \beta_{0,S} \notin C_Z(\alpha) \,\big|\, \mathcal{G}\bigr\}
\le \alpha.
\]
This proves the conditional coverage bound in the theorem.

Taking the supremum over $\beta_0$ and integrating out $\mathcal{G}$ gives the unconditional
bound
\[
\sup_{\beta_0}
P_{\beta_0}\bigl\{ \beta_{0j} \notin C_{j,n}(\alpha) \bigr\}
\le \alpha.
\]

Finally, the strong-validity property~\eqref{eq:strong-validity} for the full contour $\pi_Z(\cdot)$
follows directly from Proposition~\ref{prop:pivot-validity} applied conditional on $\mathcal{G}$,
followed by integrating over $\mathcal{G}$:
\[
\sup_{\theta_0}
P_{\theta_0}\bigl\{ \pi_Z(\theta_0) \le u \bigr\}
\le u, \qquad u \in [0,1].
\]
\end{proof}

\begin{proof}[Proof of Proposition~\ref{prop:selector-uniform}]
The result is an immediate consequence of Theorem~\ref{thm:finite-sample-strong} once we make explicit the
dependence on the selector $S$ in the underlying probability law.

Fix an arbitrary selector $S \in \mathcal{S}$ and let $P^{(S)}_{\beta_0}$ denote the law of the data
(and of the resulting RSPIM procedure) when this selector is used on the selection
subsample. By definition, $S$ is a measurable map
\[
S : (Y_{I_{\mathrm{sel}}}, X_{I_{\mathrm{sel}}}) \mapsto \widehat S \subset \{1,\dots,p\}
\]
that uses only the selection data and satisfies $|\widehat S| < n_{\mathrm{inf}}$ almost surely.
Under Assumptions (L1)–(L2), Lemma~3.1 implies that, conditional on
$G = \sigma(I_{\mathrm{sel}}, I_{\mathrm{inf}}, X, \widehat S)$, the refitted model on
$I_{\mathrm{inf}}$ is a fixed-design Gaussian linear model with support
$S = \widehat S$ and dimension $d = |S| < n_{\mathrm{inf}}$, and the $t/F$ statistics used
to construct the single-split RSPIM interval $C^{(S)}_{j,n}(\alpha)$ are exact pivots.
Importantly, the conditional distribution of these pivots given $G$ does not depend on
the particular choice of selector, as long as the selector uses only the selection data.

Applying Theorem~\ref{thm:finite-sample-strong} to the RSPIM construction based on this fixed selector $S$,
we obtain, for any coordinate $j$ and any level $1-\alpha \in (0,1)$,
\[
\sup_{\beta_0}
P^{(S)}_{\beta_0}\!\left\{\beta_{0j} \notin C^{(S)}_{j,n}(\alpha)\right\}
\;\le\; \alpha.
\]
Since this inequality holds for each fixed $S \in \mathcal{S}$, taking the supremum over
$S \in \mathcal{S}$ preserves the bound and yields
\[
\sup_{S \in \mathcal{S}} \sup_{\beta_0}
P^{(S)}_{\beta_0}\!\left\{\beta_{0j} \notin C^{(S)}_{j,n}(\alpha)\right\}
\;\le\; \alpha,
\]
which is exactly the claim of Proposition~\ref{prop:selector-uniform}.
\end{proof}

\begin{proof}[Proof of Lemma~\ref{lem:selector-pivot}]
Under Assumptions~\textup{(L1)}–\textup{(L2)}, conditioning on
$G = \sigma(I_{\mathrm{sel}}, I_{\mathrm{inf}}, X, \widehat S)$ fixes the split, the full design matrix $X$,
and the selected support $S = \widehat S$. In particular, on $I_{\mathrm{inf}}$ we have, conditional on $\mathcal{G}$,
\[
Y_{I_{\mathrm{inf}}} = X_{I_{\mathrm{inf}}}\beta_0 + \varepsilon_{I_{\mathrm{inf}}},\qquad
\varepsilon_{I_{\mathrm{inf}}} \sim N(0,\sigma^2 I_{n_{\mathrm{inf}}}),
\]
with $X_{I_{\mathrm{inf}}}$ and $S$ treated as non-random.

Writing $X_S = X_{I_{\mathrm{inf}},S}$, this is the standard fixed-design Gaussian linear model
\[
Y = X_S \beta_{0,S} + \varepsilon,\qquad \varepsilon \sim N(0,\sigma^2 I_{n_{\mathrm{inf}}}),
\]
with $d = |S| < n_{\mathrm{inf}}$. The ordinary least-squares estimator
\[
\widehat\beta_S = (X_S^\top X_S)^{-1}X_S^\top Y,\qquad
\widehat\sigma^2 = \frac{1}{\nu}\,\bigl\| Y - X_S\widehat\beta_S\bigr\|_2^2,\quad
\nu = n_{\mathrm{inf}} - d,
\]
satisfies the usual distributional identities: conditional on $\mathcal{G}$,
\[
\widehat\beta_S \sim N\bigl(\beta_{0,S},\,\sigma^2 (X_S^\top X_S)^{-1}\bigr),\qquad
\frac{\nu\,\widehat\sigma^2}{\sigma^2} \sim \chi^2_\nu,
\]
and $\widehat\beta_S$ and $\widehat\sigma^2$ are independent. Therefore, for any $j \in S$,
\[
t_j = \frac{\widehat\beta_{S,j} - \beta_{0j}}{\widehat\sigma\sqrt{v_{jj}}}
\sim t_\nu,
\]
where $v_{jj}$ is the $(j,j)$ entry of $(X_S^\top X_S)^{-1}$. Likewise, the usual likelihood-ratio
$F$-statistic for testing linear contrasts of $\beta_S$ has an $F_{d,\nu}$ distribution. These
distributions depend only on $(X_{I_{\mathrm{inf}}},S)$, which are fixed under $\mathcal{G}$, and therefore
are unaffected by the choice of selector on $I_{\mathrm{sel}}$.
\end{proof}

\subsubsection{Proofs for high-dimensional results}

The proof of Proposition~\ref{prop:hd4-lasso} is given in
Subsection~\ref{app:hd4-lasso}, where we verify Assumption~\textup{(HD4)}
for lasso stability selection under Conditions~\textup{(C1)}–\textup{(C3)}.

\begin{proof}[Proof of Theorem~\ref{thm:hd-strong}]
Let $E_n = \{S_0\subseteq\widehat S_n,\;|\widehat S_n|\le C s_n\}$. By Assumption~\textup{(HD4)},
$\inf_{\beta_0\in\mathcal{B}_{s_n}}\mathbb{P}_{\beta_0}(E_n)\to 1$. For any
$\beta_0\in\mathcal{B}_{s_n}$ with $j\in S_0(\beta_0)$ we have $j\in\widehat S_n$ on $E_n$, so
$C_{j,n}(\alpha)$ is defined on $E_n$.

Conditional on $\mathcal{G}$ and on $E_n$, the refitted model on $I_{\mathrm{inf}}$ is a fixed-dimensional
Gaussian linear model with support $S=\widehat S_n$ and $|\widehat S_n|\le C s_n$. Since
$s_n = o(n_{\mathrm{inf}})$ by Assumption~\textup{(HD1)}, there exist $\kappa\in(0,1)$ and $N<\infty$
such that $|\widehat S_n|\le \kappa n_{\mathrm{inf}}$ for all $n\ge N$ on $E_n$. Thus, for $n$ sufficiently
large, the refitted model on $I_{\mathrm{inf}}$ is low-dimensional relative to $n_{\mathrm{inf}}$ with
$d = |\widehat S_n|<n_{\mathrm{inf}}$, and the exact Gaussian $t/F$ pivots of
Lemma~\ref{lem:selector-pivot} and the finite-sample strong validity in
Theorem~\ref{thm:finite-sample-strong} apply. In particular,
\[
\mathbb{P}_{\beta_0}\{\beta_{0j}\notin C_{j,n}(\alpha)\mid \mathcal{G}\} \le \alpha
\quad\text{on }E_n.
\]

Arguing as in the finite-sample case and conditioning on $\{j\in\widehat S_n\}$, we obtain
\[
\mathbb{P}_{\beta_0}\{\beta_{0j}\notin C_{j,n}(\alpha)\mid j\in\widehat S_n\}
\le \alpha + \frac{\mathbb{P}_{\beta_0}(E_n^{\mathrm{c}})}{\mathbb{P}_{\beta_0}(j\in\widehat S_n)}.
\]
For $j\in S_0(\beta_0)$ we have $\{j\notin\widehat S_n\}\subseteq E_n^{\mathrm{c}}$, so
$\mathbb{P}_{\beta_0}(j\in\widehat S_n)\ge 1-\mathbb{P}_{\beta_0}(E_n^{\mathrm{c}})$. Taking the
supremum over $\beta_0\in\mathcal{B}_{s_n}$ with $j\in S_0(\beta_0)$ and using
$\sup_{\beta_0}\mathbb{P}_{\beta_0}(E_n^{\mathrm{c}})\to 0$ yields the desired bound. The contour
statement follows from the equivalence between non-coverage events and
$\{\pi_{n,j}(\beta_{0j})\le u\}$.
\end{proof}

\begin{proof}[Proof of Corollary~\ref{cor:uncond-strong}]
By Theorem~\ref{thm:hd-strong}, for any $u\in[0,1]$ we have
\[
\limsup_{n\to\infty}
\sup_{\beta_0\in\mathcal{B}_{s_n}: j\in S_0(\beta_0)}
P_{\beta_0}\bigl\{\pi_{n,j}(\beta_{0j}) \le u \,\big|\, j\in\widehat S_n\bigr\}
\;\le\; u.
\]
Assumption~(HD4) implies
\[
\inf_{\beta_0\in\mathcal{B}_{s_n}: j\in S_0(\beta_0)}
P_{\beta_0}(j\in\widehat S_n)\;\to\;1,
\]
so, for any such $\beta_0$,
\[
\begin{aligned}
P_{\beta_0}\bigl\{\pi_{n,j}(\beta_{0j}) \le u\bigr\}
&= P_{\beta_0}\bigl\{\pi_{n,j}(\beta_{0j}) \le u \,\big|\, j\in\widehat S_n\bigr\}
    P_{\beta_0}(j\in\widehat S_n) \\
&\quad + P_{\beta_0}\bigl\{\pi_{n,j}(\beta_{0j}) \le u \,\big|\, j\notin\widehat S_n\bigr\}
    P_{\beta_0}(j\notin\widehat S_n) \\
&\le P_{\beta_0}\bigl\{\pi_{n,j}(\beta_{0j}) \le u \,\big|\, j\in\widehat S_n\bigr\}
    + P_{\beta_0}(j\notin\widehat S_n).
\end{aligned}
\]
Taking $\sup_{\beta_0\in\mathcal{B}_{s_n}: j\in S_0(\beta_0)}$ and then $\limsup_{n\to\infty}$, the first term is bounded by $u$ by Theorem~\ref{thm:hd-strong}, and the second term tends to zero uniformly by (HD4). This yields the asserted unconditional strong-validity bound for $\pi_{n,j}(\cdot)$.

The statement for $C_{j,n}(\alpha)$ follows because $C_{j,n}(\alpha)$ is the $(1-\alpha)$ upper-level set of $\pi_{n,j}$, so that
\[
\bigl\{\beta_{0j}\notin C_{j,n}(\alpha)\bigr\}
\;\subseteq\;
\bigl\{\pi_{n,j}(\beta_{0j}) \le \alpha\bigr\},
\]
and the same argument with $u=\alpha$ gives the desired inequality.
\end{proof}

\begin{proof}[Proof of Corollary~\ref{cor:oracle}]
On the event $\{\widehat S_n = S_0\}$, $C_{j,n}(\alpha)$ and $C^{\mathrm{oracle}}_{j,n}(\alpha)$
coincide, and $\mathbb{P}(\widehat S_n = S_0)\to 1$ by selection consistency. This implies that the
ratio of diameters converges to one in probability. For the rate statement, note that on
$I_{\mathrm{inf}}$ with $S=S_0$ fixed and $|S_0|\lesssim s_n = o(n_{\mathrm{inf}})$, the OLS-based
$t$-interval has half-length proportional to $\widehat\sigma\sqrt{v_{jj}} = O_P(n_{\mathrm{inf}}^{-1/2})$
by standard linear-model theory, so both intervals have length $O_P(n_{\mathrm{inf}}^{-1/2})$.
\end{proof}

\subsubsection{Assumptions and proofs for orthogonalized extension}

\medskip
\noindent\textbf{Assumptions (O).}
\begin{itemize}
\item[(O1)] \emph{Neyman orthogonality and smoothness.}
The score $\psi(W;\theta,\eta)$ is Gateaux differentiable in $(\theta,\eta)$ and satisfies
\[
\mathbb{E}\{\psi(W;\theta_0,\eta_0)\mid\mathcal{G}\}=0,\qquad
\partial_\eta \mathbb{E}\{\psi(W;\theta_0,\eta)\mid\mathcal{G}\}\big|_{\eta=\eta_0}=0,
\]
together with suitable local smoothness in a neighborhood of $(\theta_0,\eta_0)$.

\item[(O2)] \emph{Nuisance convergence and stochastic equicontinuity.}
The cross-fitted nuisance estimators $\widehat\eta^{(-b)}$ obey
$\|\widehat\eta^{(-b)}-\eta_0\| = o_P(n_{\mathrm{inf}}^{-1/4})$ uniformly over folds
$b=1,\dots,B$, and the empirical process remainder in the expansion of $\widehat U_n(\theta)$ is
$o_P(n_{\mathrm{inf}}^{-1/2})$ uniformly over $\theta$ in a neighborhood of $\theta_0$.

\item[(O3)] \emph{Asymptotic linearity and central limit theorem.}
There exists a mean-zero, finite-variance random vector $\varphi(W;\theta_0,\eta_0)$ such that
\[
\sqrt{n_{\mathrm{inf}}}\,\widehat U_n(\theta_0)
= \frac{1}{\sqrt{n_{\mathrm{inf}}}}\sum_{i\in I_{\mathrm{inf}}}
\varphi(W_i;\theta_0,\eta_0) + o_P(1),
\]
and, conditional on $\mathcal{G}$,
\[
\frac{1}{\sqrt{n_{\mathrm{inf}}}}\sum_{i\in I_{\mathrm{inf}}}
\varphi(W_i;\theta_0,\eta_0) \;\Rightarrow\; N(0,V)
\]
for some positive-definite covariance matrix $V$.

\item[(O4)] \emph{Bootstrap approximation (uniform in $\theta$).}
Let $T_n(\theta)$ be the studentized statistic in \eqref{eq:orth-score}.
There exists a bootstrap or
wild-bootstrap analogue $T_n^*(\theta)$ such that, conditional on the data, the distribution of
$T_n^*(\theta_0)$ consistently approximates that of $T_n(\theta_0)$ \emph{uniformly} over
$\theta_0$ in the parameter set $\Theta_0$ of interest:
\[
\sup_{\theta_0 \in \Theta_0} \;\sup_{t\in\mathbb{R}}
\bigl|\mathbb{P}\{T_n(\theta_0)\le t\mid\mathcal{G}\}
      - \mathbb{P}^*\{T_n^*(\theta_0)\le t\mid\mathcal{G}\}\bigr|
\;\to\; 0 \quad\text{in probability}.
\]
\end{itemize}

\begin{proof}[Proof of Theorem~\ref{thm:orth-RSPIM}]
By Assumption~\textup{(O3)} and the consistency of $\widehat\sigma(\theta_0)$ in
\eqref{eq:orth-score}, the statistic $T_n(\theta_0)$ is asymptotically standard normal
conditionally on $\mathcal{G}$:
\[
T_n(\theta_0)
= \frac{\sqrt{n_{\mathrm{inf}}}\,\widehat U_n(\theta_0)}{\widehat\sigma(\theta_0)}
\;\xrightarrow{d}\; N(0,1)
\qquad\text{conditionally on $\mathcal{G}$}.
\]

Assumption~\textup{(O4)} states that the bootstrap calibration consistently approximates the law of
$T_n(\theta_0)$, again conditional on $\mathcal{G}$, \emph{uniformly} over $\theta_0$.
In the notation of
Proposition~\ref{prop:approx-pivot}, let $T_n(Z,\theta_0)$ denote the studentized statistic in
\eqref{eq:orth-score}, and let $F_{n,\theta_0}$ and $\widehat F_{n,\theta_0}$ be, respectively, the
true and bootstrap cumulative distribution functions of $T_n(Z,\theta_0)$ conditional on
$\mathcal{G}$.
Then (O4) implies that
\[
\sup_{\theta_0} \Delta_n(\theta_0)
:= \sup_{\theta_0} \sup_{t \in \mathbb{R}}
\bigl| F_{n,\theta_0}(t) - \widehat F_{n,\theta_0}(t) \bigr|
\;\xrightarrow{P}\; 0.
\]

Applying Proposition~\ref{prop:approx-pivot} conditional on $\mathcal{G}$, with the validified
contour constructed from $T_n(Z,\theta)$ and $\widehat F_{n,\theta}$ as in
Section~\ref{subsec:validification}, we obtain, for any $u\in[0,1]$,
\[
\sup_{\theta_0}
P_{\theta_0}\bigl\{ \pi_n^{\mathrm{orth}}(\theta_0) \le u \,\big|\, \mathcal{G}\bigr\}
\;\le\;
u + 2 \sup_{\theta_0} \Delta_n(\theta_0).
\]
Since $\sup_{\theta_0}\Delta_n(\theta_0)\xrightarrow{P}0$ by Assumption~\textup{(O4)}, taking
$\limsup_{n\to\infty}$ and then integrating out $\mathcal{G}$ yields
\[
\limsup_{n\to\infty}
\sup_{\theta_0}
P_{\theta_0}\bigl\{ \pi_n^{\mathrm{orth}}(\theta_0) \le u \bigr\}
\;\le\; u, \qquad u \in [0,1],
\]
which is the asserted asymptotic strong-validity property.
\end{proof}

\medskip
\noindent\textbf{Conditions (PL) for the partially linear model.}
\begin{itemize}
\item[(PL1)] \emph{Moment and tail conditions.}
The random vector $(Y,D,X^\top)$ has finite moments of order $q>4$ and is uniformly
square-integrable.
The conditional variance $\mathbb{E}(\varepsilon^2\mid D,X)$ is bounded away
from zero and infinity.

\item[(PL2)] \emph{Approximate sparsity and design regularity.}
There exist vectors $\gamma_0,\delta_0\in\mathbb{R}^{p_n}$ with supports of size at most $s_n$
such that
\[
g_0(X) = X^\top\gamma_0 + r_g(X),\qquad
m_0(X) = X^\top\delta_0 + r_m(X),
\]
where the approximation errors satisfy
$\mathbb{E}\{r_g(X)^2\}+\mathbb{E}\{r_m(X)^2\}=o(1)$.
The Gram matrix
$\mathbb{E}(XX^\top)$ has eigenvalues bounded away from zero and infinity on $s_n$-sparse
vectors (a restricted eigenvalue / sparse Riesz condition).

\item[(PL3)] \emph{Growth of dimension and sparsity.}
The effective sparsity $s_n$ and dimension $p_n$ may grow with $n_{\mathrm{inf}}$, but satisfy
\[
s_n^2 (\log p_n)^2 / n_{\mathrm{inf}} \to 0.
\]
\end{itemize}

The following lemma verifies that the partially linear model conditions imply the general orthogonal-score assumptions.

\begin{lemma}[Verification of Assumptions (O1)–(O4) under (PL1)–(PL3)]
\label{lem:pl-conditions}
Suppose the data follow the partially linear model
\[
Y = D\,\theta_0 + g_0(X) + \varepsilon,\qquad \mathbb{E}(\varepsilon\mid D,X)=0,
\]
with Neyman-orthogonal score
\[
\psi(Z;\theta,\eta) = \{D-m(X)\}\{Y-g(X)-\theta\{D-m(X)\}\},\qquad \eta = (g,m),
\]
and that Conditions~\textup{(PL1)}–\textup{(PL3)} hold.
Then, for the cross-fitted nuisance estimators
$\widehat\eta^{(-b)} = (\widehat g^{(-b)},\widehat m^{(-b)})$ obtained via $\ell_1$-penalized least squares,
the studentized statistic $T_n(\theta)$ in~\eqref{eq:orth-score} satisfies
Assumptions~\textup{(O1)}–\textup{(O4)} with influence function
$\varphi(W;\theta_0,\eta_0)=\psi(W;\theta_0,\eta_0)$.
\end{lemma}

\begin{proof}[Proof of Lemma~\ref{lem:pl-conditions}]
The proof proceeds in three steps, verifying each of Assumptions~(O1)–(O4) in turn.

\textbf{Step 1: Verification of (O1).}
By construction, the score $\psi(Z;\theta,\eta)$ is Gateaux differentiable in $(\theta,\eta)$.
The moment condition $\mathbb{E}\{\psi(W;\theta_0,\eta_0)\mid\mathcal{G}\}=0$ holds because
$\mathbb{E}(\varepsilon\mid D,X)=0$ and $\eta_0 = (g_0,m_0)$ are the true regression functions.
The Gateaux derivative with respect to $\eta$ at $(\theta_0,\eta_0)$ vanishes by the orthogonality
property of the score: perturbations in $(g,m)$ around $(g_0,m_0)$ do not affect the first-order
term in the expansion of $\mathbb{E}\{\psi(W;\theta_0,\eta)\mid\mathcal{G}\}$. Local smoothness
follows from the moment condition~\textup{(PL1)}.

\textbf{Step 2: Verification of (O2) and (O3).}
Under Conditions~\textup{(PL2)}–\textup{(PL3)}, the nuisance functions $g_0$ and $m_0$ admit
approximate sparse representations with effective sparsity $s_n$ and approximation error $o(1)$.
Standard results on $\ell_1$-penalized least squares \citep{MeinshausenBuhlmann2010} imply that the cross-fitted estimators $\widehat g^{(-b)}$ and
$\widehat m^{(-b)}$ achieve $L_2$ prediction error
\[
\|\widehat g^{(-b)} - g_0\|_{L_2} + \|\widehat m^{(-b)} - m_0\|_{L_2}
= O_P\bigl\{\sqrt{s_n \log p_n / n_{\mathrm{inf}}}\bigr\}
= o_P(n_{\mathrm{inf}}^{-1/4}),
\]
uniformly over folds $b=1,\dots,B$, where the last equality uses Condition~\textup{(PL3)}.
This verifies the nuisance-rate condition in~\textup{(O2)}.

By the Neyman orthogonality established in Step 1, the effect of nuisance estimation on the
cross-fitted score is of order $o_P(n_{\mathrm{inf}}^{-1/2})$, yielding the stochastic
equicontinuity required in~\textup{(O2)}. Combined with a central limit theorem for the
empirical average of the influence function $\varphi(W;\theta_0,\eta_0) = \psi(W;\theta_0,\eta_0)$,
this establishes the asymptotically linear representation and conditional CLT in~\textup{(O3)}.
The conditional covariance $V$ in~\textup{(O3)} is positive-definite under the non-degeneracy
assumptions in~\textup{(PL1)}.

\textbf{Step 3: Verification of (O4).}
Under the moment condition~\textup{(PL1)}, a multiplier wild bootstrap with mean-zero,
unit-variance multipliers (e.g., Rademacher) consistently approximates the conditional
distribution of the studentized statistic $T_n(\theta_0)$ in~\eqref{eq:orth-score},
uniformly over $\theta_0$ in a compact parameter set. This follows from standard wild-bootstrap
theory for $U$-statistics and sample averages \citep{mammen1993bootstrap}.
The uniformity in $\theta_0$ required in~\textup{(O4)} holds by the stochastic
equicontinuity established in Step 2 and a bracketing or chaining argument.

Combining these three steps yields the claimed verification of Assumptions~\textup{(O1)}–\textup{(O4)}.
\end{proof}

\begin{proof}[Proof of Corollary~\ref{cor:pl-orth}]
By Lemma~\ref{lem:pl-conditions}, Conditions~\textup{(PL1)}–\textup{(PL3)}
imply Assumptions~\textup{(O1)}–\textup{(O4)} for the studentized statistic
$T_n(\theta)$ in~\eqref{eq:orth-score}.
Applying
Theorem~\ref{thm:orth-RSPIM} with this statistic yields the claimed
asymptotic strong-validity property.
\end{proof}

\subsubsection{Proof for multi-split aggregation}

\begin{proof}[Proof of Proposition~\ref{prop:maxitive}]
Fix $\theta_0$ and $u\in[0,1]$. By construction,
\[
\{\pi^{\max}_{n,j}(\theta_0)\le u\}
= \Bigl\{\max_{1\le r\le R}\pi^{(r)}_{n,j}(\theta_0)\le u\Bigr\}
= \bigcap_{r=1}^R \{\pi^{(r)}_{n,j}(\theta_0)\le u\}.
\]
Hence
\[
\mathbb{P}_{\theta_0}\{\pi^{\max}_{n,j}(\theta_0)\le u\}
= \mathbb{P}_{\theta_0}\Bigl(\bigcap_{r=1}^R \{\pi^{(r)}_{n,j}(\theta_0)\le u\}\Bigr)
\le \min_{1\le r\le R}
\mathbb{P}_{\theta_0}\{\pi^{(r)}_{n,j}(\theta_0)\le u\}.
\]
Strong validity of each $\pi_n^{(r)}$ implies
$\mathbb{P}_{\theta_0}\{\pi^{(r)}_{n,j}(\theta_0)\le u\}\le u$ for all $r$ and all $u\in[0,1]$, so
\[
\mathbb{P}_{\theta_0}\{\pi^{\max}_{n,j}(\theta_0)\le u\} \le u.
\]
Taking the supremum over $\theta_0$ yields strong validity of $\pi^{\max}_{n,j}$.

For the coverage statement, note that
\[
\{\beta_{0j}\notin C^{\max}_{j,n}(\alpha)\}
= \bigcap_{r=1}^R \{\beta_{0j}\notin C^{(r)}_{j,n}(\alpha)\}
\subseteq \{\pi^{\max}_{n,j}(\beta_{0j})\le \alpha\}.
\]
Thus
\[
\mathbb{P}_{\theta_0}\{\beta_{0j}\notin C^{\max}_{j,n}(\alpha)\}
\le \mathbb{P}_{\theta_0}\{\pi^{\max}_{n,j}(\beta_{0j})\le \alpha\} \le \alpha,
\]
and hence $\mathbb{P}_{\theta_0}\{\beta_{0j}\in C^{\max}_{j,n}(\alpha)\}\ge 1-\alpha$.
\end{proof}

\subsection{Implementation notes}

We cap the size of the selected set by $k_{\max}\le n_{\mathrm{inf}}$ to ensure that least-squares refits on
$I_{\mathrm{inf}}$ remain well-posed.
Wild bootstrap implementations use Rademacher multipliers on the residuals.
Carving variants are implemented by reusing a fraction of the selection sample for refitting; these are flagged
as approximate, since the overlap between selection and inference samples breaks the exact finite-sample pivot
calibration.
\section{Supplementary experimental details}\label{app:experiments}

\subsection{Detailed experimental designs}\label{app:designs}

We consider high-dimensional linear models with Gaussian features and correlation $\rho\in\{0,0.5\}$, signal sparsity $s\in\{5,10\}$, sample sizes $n\in\{100,200\}$, and ambient dimensions $p\in\{500,2000\}$. Signal-to-noise ratio (SNR) is controlled via $\|\beta\|_2$. For robustness, we introduce (i) heteroskedastic noise with $\sigma_i^2\propto X_{i1}^2$ and (ii) heavy-tailed noise (e.g., $t_3/t_4$). To avoid over-selection, we cap the selected set size by
$|\widehat S|\le k_{\max}=\lfloor 0.5\,n_{\mathrm{inf}}\rfloor$,
where $n_{\mathrm{inf}}$ denotes the effective degrees of freedom available to inference.

In addition, to stress-test post-selection stability under ill-conditioned selection events, we include a moderate-dimensional but highly correlated design with $(n,p,\rho)=(100,200,0.9)$ and weak signals ($\beta_{0j}\approx 0.4$ on $S_0$). This configuration is used in Module~E to compare RSPIM with the polyhedral exact selective-inference method.

\subsection{Calibration diagnostics: IM/PIM-standard validity checks}\label{app:calibration-diagnostics}

We report three standard diagnostics for assessing the validity of our plausibility contours and intervals:

\begin{enumerate}[leftmargin=*,itemsep=0.25em]
\item \textbf{Strong validity.} For true coordinates, the null \emph{plausibility} (equivalently, the $p$-value induced by our pivot) should be close to $\mathrm{Unif}(0,1)$ or conservative. We display histograms and QQ plots, and summarize deviations via Kolmogorov distances.

\item \textbf{Coverage vs nominal.} We compare empirical coverage of $(1-\alpha)$ plausibility intervals for $\beta_j$ against the nominal level over $\alpha\in\{0.1,\dots,0.9\}$, both per-coordinate and averaged over coordinates.

\item \textbf{False-confidence check.} For selected \emph{null} coordinates, we report the rate at which plausibility exceeds $1-\alpha$ to verify the expected control ($\le \alpha$ in practice, hence conservative).
\end{enumerate}

Beyond these scalar summaries, we also visualize the plausibility contour itself for representative coordinates. The plausibility contour provides a genuinely possibilistic diagnostic: maxitive aggregation trades local sharpness for ``strong but diffuse'' evidence, making it easy to detect when repeated splits lead to diffuse post-selection uncertainty.

\subsection{Efficiency comparison methodology}\label{app:efficiency-methodology}

All efficiency comparisons are made \emph{at equal conditional coverage}. Formally, for each method $m$ (RSPIM, de-biased lasso, etc.) and each coordinate $j$ we can regard the intervals as arising from a pivotal statistic $T_{n,j}^{(m)}(\beta_j)$ and an associated plausibility contour or $p$-value function. To compare efficiency fairly, we introduce a one-parameter family of ``shrunk'' pivots
\[
T_{n,j}^{(m,c)}(\beta_j) = c\,T_{n,j}^{(m)}(\beta_j), \qquad c>0,
\]
and let $C_{j,n}^{(m,c)}(\alpha)$ denote the corresponding $(1-\alpha)$ interval for $\beta_j$. For a fixed design and method $m$, we then define the empirical conditional coverage at $c$ (the Monte Carlo analogue of the conditional-on-selection coverage in Theorem~\ref{thm:hd-strong}) as
\[
\widehat{\mathrm{Cov}}^{\mathrm{cond}}_m(c)
= \frac{1}{N_{\mathrm{sel},m}}
  \sum_{b=1}^B \sum_{j\in\widehat S_m^{(b)}}
  \mathbf{1}\bigl\{\beta_{0j}\in C_{j,n}^{(m,c)}(\alpha; b)\bigr\},
\]
where $b=1,\dots,B$ indexes Monte Carlo repetitions, $\widehat S_m^{(b)}$ is the selected set for method $m$ in repetition $b$, and $N_{\mathrm{sel},m}=\sum_{b=1}^B|\widehat S_m^{(b)}|$ is the total number of selected coordinates across repetitions.

For each method $m$ we search over a grid of $c$-values and choose a calibration $c_m^\star$ such that $\widehat{\mathrm{Cov}}^{\mathrm{cond}}_m(c_m^\star)$ is close to the target level $1-\alpha$ (here $1-\alpha=0.90$). We then report, at these calibrated values $c_m^\star$, (i) the mean interval length (MIL) across selected coordinates, and (ii) the empirical rejection probability of the null at nominal level $\alpha$.

This yields ``equal-coverage'' efficiency comparisons: every method is operated at approximately the same conditional coverage, and differences in MIL or power reflect efficiency rather than miscalibration.

Importantly, the RSPIM procedure defined in Sections~\ref{sec:method}--\ref{sec:theory} corresponds to the uncalibrated case $c=1$. All strong-validity and oracle-equivalence results are stated for this raw version. The calibration factors $c_m^\star$ in this section should therefore be interpreted purely as simulation diagnostics: they quantify how much each method would need to be shrunk or expanded, in a given design, to attain the target coverage, and they are \emph{not} part of the formal methodology.

\subsection{Additional remarks and deferrals}\label{app:additional-remarks}

Comparisons with alternative selectors, including knockoffs and forward stepwise, as well as detailed results for intersection-of-splits intervals, are reported in the supplement. In robust settings we emphasize the calibrated single-split RSPIM as the main estimator and use the union-of-splits aggregator as a conservative stability device, since in some stress designs the union produces no effective intervals ($N=0$ at the calibrated $c$).

\end{document}